\pdfoutput=1

\documentclass[a4paper]{article}   %[twocolumn]   

\usepackage{graphicx}% Include figure files
\usepackage{dcolumn}% Align table columns on decimal point
\usepackage{bm}% bold math
\usepackage[utf8]{inputenc}
\usepackage[T1]{fontenc}

\usepackage{mathptmx} %times

\usepackage{cite} %\cite{...,...,...} gives [1-3] instead of [1], [2], [3]

\usepackage{amsmath}
\usepackage{amssymb}
\usepackage{mathrsfs}   %\mathscr A
\usepackage[caption=false]{subfig}
\usepackage{color}  		%color
\usepackage{tabularx} %tablularx
\usepackage{wasysym} %symbol astrosun
\usepackage{layouts} %output of textwidth and linewidth

\usepackage{amsthm}
\usepackage[capitalize,nameinlink,noabbrev]{cleveref} %!!!!!!!!!!!!!!
\theoremstyle{definition}

\newtheoremstyle{theoremstyle1}% <name>
{10pt}% <Space above>
{17pt}% <Space below>
{}% <Body font>
{}% <Indent amount>
{\itshape}% <Theorem head font> %{\itshape} {\bfseries}
{}% <Punctuation after theorem head>
{\newline}% <Space after theorem head>
{\thmname{#1}\thmnumber{ #2}\thmnote{ (#3)}}% <Theorem head spec (can be left empty, meaning `normal')>

\theoremstyle{theoremstyle1}
\newtheorem{definition}{Definition}[section]

\newtheorem{remark}{Remark}[section]
\newtheorem{example}{Example}[section]

\newtheorem{proposition}{Proposition}[section]

 	\definecolor{gray}{rgb}{0.7,0.7,0.7}
	
\newlength{\LL}
 \newcommand{\rev}[1]{#1}
 \newcommand{\revi}[1]{{#1}}

\begin{document}

\title{On Norm-Based Estimations for
\\
 Domains of Attraction in \\  Nonlinear Time-Delay Systems
}
\date{}
\author{Tessina H. Scholl
\thanks{
Institute for Automation and Applied Informatics,
Karlsruhe Institute of Technology (KIT), 
76344 Eggenstein-Leopoldshafen, Germany,\newline
e-mail: \{tessina.scholl; veit.hagenmeyer; lutz.groell\}@kit.edu 
} 
    \and
Veit Hagenmeyer\footnotemark[1]\and
Lutz Gr\"oll \footnotemark[1]
}

\maketitle

\begin{abstract}
For nonlinear time\rev{-}delay systems, domains of attraction are rarely studied despite their importance for technological applications. The present paper  \rev{provides methodological hints for the determination of} an upper bound on the radius of attraction by \rev{numerical means}. Thereby, the respective Banach space for initial functions has to be selected and primary initial functions have to be chosen. The latter are used in time\rev{-}forward simulations to determine a first upper bound on the radius of attraction. Thereafter, this upper bound is refined by secondary initial functions\rev{,} which result a posteriori from the preceding simulations. Additionally, a bifurcation analysis \rev{should be} undertaken.
 %with the time delay as bifurcation parameter. 
This analysis results in a possible improvement of the previous estimation.
An example of a time-delayed swing equation demonstrates \rev{the various aspects}.
\end{abstract}
%\printinunitsof{cm}\prntlen{\textwidth}
%\printinunitsof{cm}\prntlen{\linewidth}
%%

\section{\label{sec:Introduction}Introduction}

Time delays due to communication, measurement, data processing, delayed actuator reactions\rev{,} or transport processes are omnipresent in technical systems. That is why stability analysis of time\rev{-}delay systems has gathered increasing interest in the last decades. However, the focus was mainly on stability criteria for linear systems  \cite{Gu.2003,Niculescu.2004,Briat.2015,Wu.2010,Michiels.2014}, be it Lyapunov-Krasovskii functionals \cite{Hale.1993,Seuret.2016,Fridman.2014b}, Lyapunov-Razumikhin functions \cite{Hale.1993}, comparison principles \cite{Niculescu.2001,Halanay.1966,Dambrine.1994}, input-output ap\-proach\-es \cite{Briat.2011}\rev{,} or eigenvalue calculations \cite{Breda.2015,Engelborghs.2002,Insperger.2017,Jarlebring.2008}. \rev{By contrast, in technical applications, we encounter nonlinear systems \cite{Otto.2019,Dombovari.2008,Yan.2019,Dombovari.2019,Schafer.2015}}.
Indeed, for a nonlinear system the Principle of Linearized Stability \cite{Diekmann.1995} allows to deduce stability or instability of equilibria, but the question arises: what is the practical relevance of knowledge about asymptotic stability 
if there is no knowledge about the domain of attraction? 
In this light, small initial disturbances can still result in departing trajectories. In relation to the results listed above, estimations on  domains of attraction in time\rev{-}delay systems are %still 
rare in literature \cite{MelchorAguilar.2006,Villafuerte.2007,Fridman.2014,Cao.2002,Coutinho.2008}. 
Time\rev{-}delay systems require an initial function instead of the initial value, which would be sufficient in ODEs. Hence, the state space, 
which contains the domain of attraction,
%as well as the domain of attraction as a part of it, 
is infinite dimensional.
The present paper addresses such domains of attraction for equilibria in autonomous retarded functional differential equations (RFDE) \revi{with a focus on systems with a discrete time delay}. 
\\
\\
To this end, consider the spectral abscissa, %(largest eigenvalue of the infinitesimal generator)
which \rev{in an eigen\-value-based stability analysis of the linearized system} decides about asymptotic stability and instability ({\itshape Principle of Linearized Stability}, Diekmann \cite{Diekmann.1995}\rev{, Ch. VII, Theorem 6.8}). It should be noted that this decisive number has neither an implication to allowed perturbations of the RFDE right-hand side or system parameters, i.e.\ robustness, nor to allowed initial perturbations, i.e.\ the domain of attraction \cite{Goldsztejn.2019}. 
However, in practical applications various requirements concerning the domain of attraction occur and necessitate adequate approaches.
\par
{\begingroup
\parskip0.3\baselineskip
\setlength\parindent{-2.1em} %to the left in first line
\leftskip=2.1em %to the right overall
\par
{\itshape (R1) }{\itshape (test)} Let a certain initial function be given. Is this initial function an element of the domain of attraction?
(e.g.\ contingency tests in power systems) 
\par
{\itshape (R2) }{\itshape  (prove)}  Let a bound for the initial functions in a certain norm be given. Do all initial perturbations within this bound lead to attractor-convergent solutions? (e.g.\ requirements in controller design) 
\par
{\itshape (R3) }{\itshape (disprove)} Proving non-fulfillment of (R2). 
\par
{\itshape (R4) }{\itshape (compare)} Let different systems or controller configurations be given. Which of the systems is the better one w.r.t.\ the domain of attraction in the sense of some order relation? 
\par
{\itshape (R5) }{\itshape (parametrize)} Let a system or controller with  free parameters be given. How is the domain of attraction influenced by the parameters? (e.g.\ system analysis and controller synthesis)
\par
\endgroup}
\par
\noindent Different approaches are possible in order to meet the requirements listed above.
\\[\LL]
{\itshape (i) A single time\rev{-}domain simulation} for each specific initial function of interest is able to fulfill requirement (R1) in a hardly conservative way. However, in contrast to norm-based criteria, which - once formulated - allow to conclude convergence directly, the computational effort of such individual simulations should be kept in mind. This is for instance the reason why in (delay-free) power system stability analysis, methods have been introduced,  which do no longer rely only on on-demand time\rev{-}domain simulations \cite{Chiang.2011,Chiang.2009}. 
\\[\LL]
{\itshape (ii) Numerical Basin stability} \cite{Menck.2013} in its original form for time-delay-free systems approximates the Lebesgue measure of the finite dimensional domain of attraction (blue shaded area in Figure \ref{fig:Measures}). 
\begin{figure} 
\centering
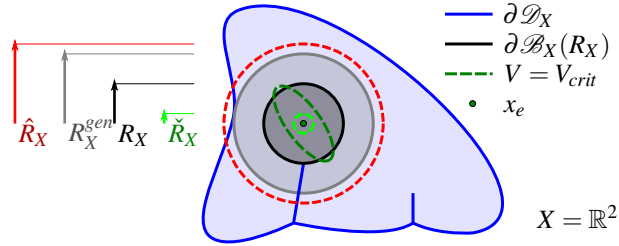
\caption{\label{fig:Measures} Discussed \revi{numbers for} the domain of attraction $\mathcal D_X$ of an equilibrium point $x_e$ and their relation to the radius of attraction $R_X$. For the sake of visualization, the finite dimensional case $X=\mathbb R^2$ is considered. 
}
\end{figure}
Hence, there is a relation to the probability of randomly chosen initial values to be within the domain of attraction. Leng et al.\ \cite{Leng.2016} give a generalization to time\rev{-}delay systems. The derived number can contribute to a comparison of different system configurations as required in (R4).
\\[\LL]
{\itshape (iii) The level set of a Lyapunov-Krasovskii functional} \cite{Hale.1993,MelchorAguilar.2006} can be used as estimation for the domain of attraction (represented by the inner of $V\rev{=} V_{crit}$ in Figure \ref{fig:Measures}). Consequently, it yields a scalar criterion in order to test whether an initial function surely belongs to the domain of attraction (R1). The background is a generalization of LaSalle's invariance principle to RFDEs \cite{Hale.1965b}. However, already for linear systems Lyapunov-Krasovskii functionals are challenging. 
Vast efforts have been made in view of the possible delay dependence of stability in such linear RFDEs. The challenge is then to derive functionals that conclude stability even for time delay values near the bound at which stability is lost. Nonlinear systems necessitate individual approaches. Alternatively, instead of searching for such special functionals, there is also the ansatz to use a functional \rev{that} is suitable for the linearized system. The latter approach is well-known for ordinary differential equations (ODE) \cite{Khalil.2002}. \revi{However, it has} to be reckoned with conservative estimations. Lyapunov-Krasovskii functionals are in general complicated expressions, such that the raw sublevel set criterion itself is not very intuitive. Therefore, it is usually translated to a norm criterion, see (vi). 
\\[\LL]
{\itshape (iv) The level set of a Lyapunov-Razumikhin function} \cite{Souza.2014b} can be used as well, since there are also Razumikhin-based generalizations of LaSalle's invariance principle \cite{Haddock.1983}. %However, the same discussion as in (iii) holds.
\\[\LL]
{\itshape (v) Bounds in a finite dimensional parameter space} \cite{Liu.2014}, by analytical or numerical methods, come into play if only a certain class of parametrizable initial functions is relevant. For instance, a jump from a predefined initialization to a variable value is a class of initial functions \rev{that} is fully determined by a single parameter. % and hence finite dimensional. 
Then, the space of initial functions is indeed finite dimensional and the finite number of parameters is decisive. 
\\[\LL]
{\itshape (vi) 
An underestimation of the radius of attraction} yields a norm criterion \rev{that} ensures convergence to the attractor (interior of the green dashed circle with radius $\check R_X$ in Figure \ref{fig:Measures}). It is essential in (R2) and yields a very simple criterion for testing in (R1). Of course, usually nothing is known about the conservativeness of such criteria. 
Results based on Lyapunov-Krasovskii functionals can be translated to such easier to handle norm criteria, which goes along with further conservativeness \cite{MelchorAguilar.2006,Villafuerte.2007}.
\\[\LL]
{\itshape (vii)  An overestimation of the radius of attraction} (interior of the red dashed circle with radius $\hat R_X$ in Figure \ref{fig:Measures}) means an upper bound on possible norm criteria described above. Such an overestimation is needed in order to prove that any larger norm requirement certainly cannot be fulfilled, i.e.\ (R3). \rev{Although} no real relations between domains of attraction can be deduced, the knowledge that one system surely possesses a small domain of attraction might also give hints for comparison purposes (R4). \rev{Furthermore, there is usually no means of getting an impression how close a lower bound might be to the real radius of attraction, unless an upper bound is provided additionally.}
As soon as an initial state \rev{that} does not belong to the domain of attraction is found, its norm is known to be such an upper bound on the radius of attraction.
This is what simulative approaches are predestined for and what the present paper addresses.   
\\
\\
The present paper is organized \revi{as follows.} The introduction closes with used notations and needed preliminaries in the context of autonomous retarded functional differential equations. Section \ref{sec:DomainOfAttraction} addresses domain and radius of attraction as well as generalized concepts thereof. Section \ref{sec:SelectionOfTheBanachSpace} shows why the state space selection should be well considered and proposes how the frequent situation of differently delayed state variables can be appropriately treated. Section \ref{sec:SimApproach} provides the reader with instructions on how a simulation\rev{-}based approach should be tackled. In section \ref{sec:BifApproach}, a bifurcation analysis\rev{-}based approach is introduced, which results in a possible refinement of previous estimations. Finally, the overall methodology is demonstrated with the delayed swing equation as example system. 

\subsection{\label{sec:}Notation and Preliminaries}
Consider the mathematical nomenclature:\\[0.5em]
\begin{tabular}{ll}
$C([a,b],\mathbb R^n)$ &space of continuous functions with\\
&  $\|\phi \|_C\stackrel{\textrm{def}}{=}\max\limits_{\theta\in [a,b]}(\|\phi(\theta)\|_2)$ \\
$PC([a,b],\mathbb R^n)$ &space of piecewise continuous functions,\\
& $\|\cdot\|_{PC}=\|\cdot\|_{C}$ \\
%$C^k([a,b],\mathbb R^n)$ &$k$-times continuously differentiable functions \\
%$AC([a,b],\mathbb R^n)$ &absolutely continuous functions \\
%$Lip_{loc}([a,b],\mathbb R^n)$	& locally Lipschitz continuous functions\\
%$\|\cdot \|_C$	& uniform norm in $C([a,b],\mathbb R^n)$ \\ %or $PC([a,b],\mathbb R^n)$\\
%$DF(x)$ &Fréchet derivative  \\
$L^p([a,b],\mathbb R^n)$ & $L^p$-space with \\
&$\|x \|_{L^p}=(\int_a^b \|x(\theta)\|^p\textrm d\theta)^{1/p}$ \\
%$\Re e $, $\Im m$  &real, imaginary part \\
%$\sigma(\mathscr A)$, $\sigma_P(\mathscr A)$ & (point) spectrum of an operator $\mathcal A$ \\
%$\mathbb C$, $\mathbb C_+$ & complex numbers (with positive real part)\\
$\mathbb R$, $\mathbb R_{> 0}$ & real (positive) numbers \\
$\mathcal  B_X(c;r)$ &open ball in a normed space $X$\\
& with center $c\in X$ and radius $r\in \mathbb R_{>0}$ \\
$\mathcal  B_X(r)$ &$=\mathcal  B_X(0_X;r)$ \\
%$0_X$ &zero element of the normed space $X$ \\
$0_{[a,b]}$; $0_n$ &zero function $[a,b] \to \mathbb R^n$; zero in $\mathbb R^n$ \\
$\chi_{M}(t)$ &indicator function w.r.t. a set $M$ \\
& $\chi_{M}(t)=\left\{\begin{array} {ll}1, & \textrm{if } t\in M\\ 0, & \textrm{otherwise }  \end{array} \right.$ \\
%$\textrm{id}$ 	&identity operator \\
$\textrm{dom}$ 	&domain of a function \\
$\partial M$, $\overline M$ 	&boundary, closure of a set $M$ \\
$x(t;\phi)$	&solution at time $t$ for initial function $\phi$ \\
$x_t=x_t(\cdot;\phi)$	&state, i.e.\ delay-width segment of $x$ at  $t$ \\
&$x_t\colon[-\tau,0] \to \mathbb R^n$,  $\theta \mapsto x_t(\theta):=x(t+\theta)$ 
%\\
%$\mathscr T(t)$ &solution operator \\
%&$\mathscr T(t):X \to X$, $\phi \mapsto \mathscr T(t)\phi:=x_t(\cdot;\phi)$
\end{tabular}
\\
\\
In the following, we revisit some required preliminaries concerning {\itshape autonomous retarded functional differential equations (RFDE)}. % $\dot x(t)&=f(x_t)$, $x(t)\in \mathbb R^n$, $f\colon X\to \mathbb R^n$.
Subsequently, $X$ denotes a normed space on the delay interval $[-\tau,0]$, $\tau\in \mathbb R_{>0}$, e.g., the Banach space of continuous functions $X=(C([-\tau,0],\mathbb R^n),\|\cdot\|_C)$ endowed with the uniform norm $\|\phi\|_C\stackrel{\textrm{def}}{=}\max\limits_{\theta\in [-\tau,0]}(\|\phi(\theta)\|_2)$.
The notation of a state as $x_t\in X$,
\begin{align}
x_t(\theta)\stackrel{\textrm{def}}=x(t+\theta), \quad \theta \in [-\tau,0]\rev{,}
\label{eq:notation}
\end{align}
is common.
It represents a restriction of the solution function $x(\tilde t)$ to the time horizon $\tilde t\in[t-\tau,t]$ (Figure \ref{fig:states_3d}).
\begin{figure}
\centering
\includegraphics{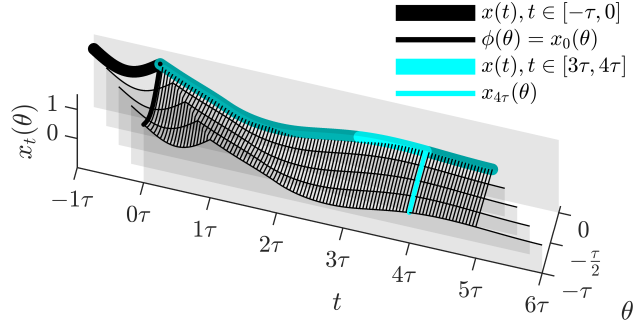}
\caption{\label{fig:states_3d} Exemplary initial data $\phi\in X$ and state $x_{4\tau}\in X$ % (solution for (\ref{eq:scalarExample}) with $\phi(\theta)=1+0.5\sin(\pi\theta)$).
}
\end{figure}
Hence, a {\itshape forward trajectory} of an initial function $x_0=\phi$
\begin{align}
\gamma^+(\phi)\stackrel{\textrm{def}}=\{x_t(\cdot;\phi)\in X:t\in[0,t_{max})\}
\end{align}
describes the set of all states \rev{that} arise over the whole forward time in which the solution exists $t\in[0,t_{max})$. Thereby, $x_t(\cdot;\phi)$ refers to
solutions of the Cauchy problem 
\begin{alignat}{4}
\dot x(t)&=F(x_t) ,\quad &&t>0 
\label{eq:RFDE}\\
x(\theta)&=\phi(\theta), &&\theta\in[-\tau,0] \nonumber
\end{alignat}
with an initial state $\phi\in X$ and $x(t)\in \mathbb R^n$, $\dot{()}=\frac{\textrm{d}(\cdot)}{\textrm{d}t}$, $F\colon X\to \mathbb R^n$ continuously differentiable, $F(0_{[-\tau,0]})=0_n$. 
For theorems about well-posedness in certain Banach spaces\rev{,} see, e.g.,  Hale and Verduyn Lunel \cite{Hale.1993} or Kharitonov \cite{Kharitonov.2013}.
We are mainly interested in {\itshape autonomous differential difference equations} \cite{Bellman.1963} with a single discrete delay $\tau>0$, as a special type of RFDE (\ref{eq:RFDE}), where
\begin{align}
F(x_t)=f(x(t),x(t-\tau)),
\label{eq:RFDE_singleDelay}
\end{align} 
$f\colon \mathbb R^n \times \mathbb R^n \to \mathbb R^n$. 
\noindent An equilibrium state of (\ref{eq:RFDE}) is defined as $\phi_e(\theta)\equiv x_e: x_t(\cdot;\phi_e) =\phi_e(\cdot)$, $\forall t \geq 0$ with an equilibrium solution $x(t;\phi_e) \equiv x_e$. The value $x_e$ is simply named equilibrium. W.l.o.g. we assume $x_e=0_n$ for the equilibrium of interest.
The focus of the present paper lies on the domain of attraction of this equilibrium. Hence, the trivial equilibrium is required to be attractive. Although it is well known that attractivity and stability according to Lyapunov are actually independent properties in nonlinear systems \cite{Hahn.1967}, we are only concerned with their combination, i.e.\ asymptotic stability. The latter has the advantage that it can easily be proven by the {\itshape{Principle of Linearized Stability}} (Diekmann et al. \cite{Diekmann.1995}, Chapter VII, Theorem 6.8) if the equilibrium is hyperbolic.
\begin{definition}[Local Asymptotic Stability (LAS) in $X$]
\label{def:LAS}
The zero equilibrium of the autonomous system (\ref{eq:RFDE}) is 
{\itshape locally asymptotically stable  (LAS) in $X$} if  \\
it is {\itshape stable} in $X$, i.e.\ 
\begin{align}
\forall \varepsilon>0, \;\exists \delta(\varepsilon)>0: \; \|\phi\|_X\leq\delta  \; \Rightarrow \|x_t\|_X \leq \varepsilon, \; \forall t\geq 0\rev{,}
\label{eq:stable}
\end{align}
and {\itshape locally attractive} in $X$, i.e.\ 
\begin{align}
\exists \delta_a>0: \quad \|\phi\|_X\leq \delta_a  \; \Rightarrow \lim_{t\to \infty} \|x_t\|_X=0.
\label{eq:attractive}
\end{align}
\end{definition}
\noindent Definition \ref{def:LAS} implicitly assumes forward completeness of solutions in an environment around the zero solution, i.e.\ the solutions under consideration exist for all $t\in[0,\infty)$. 
%This is ensured by continuous dependence on initial data. 
%%%%%%%%%%%%%%%%%%%%%%%%%%%%%%%%%%%%%%%%%%%%%%%%%%%%%%%%%
%%%%%%%%%%%%%%%%%%%%%%%%%%%%%%%%%%%%%%%%%%%%%%%%%%%%%%%%%
%%%%%%%%%%%%%%%%%%%%%%%%%%%%%%%%%%%%%%%%%%%%%%%%%%%%%%%%%%
\section{Domain and Radius of Attraction} \label{sec:DomainOfAttraction}
\subsection{Domain of Attraction}
The present paper deals with the domain of attraction $\mathcal D_{X}\subseteq X$ (also referred to as  region of attraction,   basin of attraction\rev{,} or region of asymptotic stability in literature) of an asymptotically stable equilibrium state in $\phi_e=0_{[-\tau,0]}$.
Thus, the domain of attraction collects all initial functions $\phi\in X$ that lead to a zero-convergent solution. 
\begin{definition}[Domain of Attraction]
The {\itshape domain of attraction} of an asymptotically stable zero equilibrium is defined as
\begin{align}
\mathcal D_{X}=\{\phi\in X: x_t \text{ exists on } t\in[0,\infty) \text{ and }\lim_{t\to \infty}\|x_t\|_X = 0  \}. \nonumber
\end{align}
\end{definition}
\subsection{Radius of Attraction} 
Based on the notion of a norm ball around $0_{[-\tau,0]}=\phi_0\in X$\rev{,}
\begin{align}
\mathcal  B_X(r)&:=\mathcal  B_X(0_{[-\tau,0]};r)\\
\textrm{with } \mathcal  B_X(\phi_0;r)&\stackrel{\textrm{def}}{=}\{\phi\in X: \|\phi-\phi_0\|_{X}< r\}\rev{,}
\end{align}
a strong simplification can be derived:
\begin{definition}[Radius of Attraction and Ball of Attraction]
The {\itshape radius of attraction} in $X$ of an asymptotically stable equilibrium is defined as
\begin{align}
R_X=\sup\{r > 0: \mathcal B_X(r)\subseteq \mathcal D_{X}\}
\label{eq:RadiusOfAttraction}
\end{align}
and the {\itshape ball of attraction} $\mathcal B_X(R_X)$ denotes the largest norm ball inside the domain of attraction $\mathcal D_X$. 
\end{definition}
\noindent The radius of attraction is schematically represented by the black circle in Figure \ref{fig:Measures}. Thereby, the definition of attractivity in (\ref{eq:attractive}) ensures $\mathcal B(R_X)\neq \emptyset$ by $R_X\geq\delta_a>0$. Of course, the exact value of $\mathcal B(R_X)$ remains usually unknown.
\begin{remark}[Notion of the Radius of Attraction]
The term radius of attraction is analogously used in the context of forward / pullback attractors in nonautonomous delay-free systems \cite{Kloeden.2011}. The definition is in accordance with the well-known universal definition of a stability radius as largest ball, whose elements satisfy a certain stability condition. However, it should be carefully distinguished between the question of robustness, i.e.\ parameter perturbations, which is often addressed by this concept  \cite{Hinrichsen.1990b,Hu.2003}, and the domain of attraction, i.e.\ initial value perturbations.   
\end{remark}
\noindent While the radius of attraction (\ref{eq:RadiusOfAttraction})  is defined as largest ball in the interior of the domain of attraction, it can also be considered as the minimum distance of the complement $\mathcal D_X^c$ to the origin. This leads to the equivalent definition 
$
R_X=\inf_{\phi\in X}\{\|\phi\|_X: x(t;\phi) \not \to 0\}
$, which is helpful for numerical estimations.
\begin{definition}[Upper Bound on the Radius of Attraction Based on a Set of Initial Functions] \label{def:hat_R_X}
Assume $\Phi \subset X$ is a set of tested initial functions.
The resulting {\itshape upper bound on the radius of attraction} $\revi{\hat R_X(\Phi)}\geq R_X$ is given by
\begin{align}
\revi{\hat R_X(\Phi)}=\inf_{\phi\in \Phi}\{\|\phi\|_X: x(t;\phi) \not \to 0\}.
\label{eq:hat_R_X}
\end{align}
\end{definition} 
\noindent If no divergent solution can be found, the above definition yields 
\begin{align}
\revi{\hat R_X(\Phi)}=\inf \;\emptyset = \infty.
\end{align}
The following simple scalar example demonstrates how sensitive such estimations are to the choice of initial functions.
\begin{example}[Scalar Nonlinear RFDE] 
\label{ex:scalarExample}
 Consider 
\begin{align}
\label{eq:scalarExample}
\dot x (t)&=-x(t)-x(t-\tau)+x^3(t), \quad &t>0& \\
x(\theta)&=\phi(\theta),&\theta\in[-\tau,0]& \nonumber
\end{align}
with $x(t)\in \mathbb R$ and $\phi \in PC([-\tau,0],\mathbb R)$.
Figure \ref{fig:scalarExample} shows some solutions, where initial functions in the domain of attraction $\phi\in \mathcal D_{PC}$ are printed in black.
\begin{figure*} 
\ifdim 17cm<\textwidth %\linewidth
%!!!!!!!!!!!!!!!!!!!!!!!!!!!!!!!!!!!!!!!!!!!!!!!!!!!!!!!!!!!!!!!!!!!!!
%\begin{subfigure}[b]{0.15\textwidth}
%\subfloat[]
\begin{minipage}[b]{1.5cm}
$\tau=1:$
\vspace{9.5em}
\end{minipage}
 \subfloat[]{
\includegraphics[]{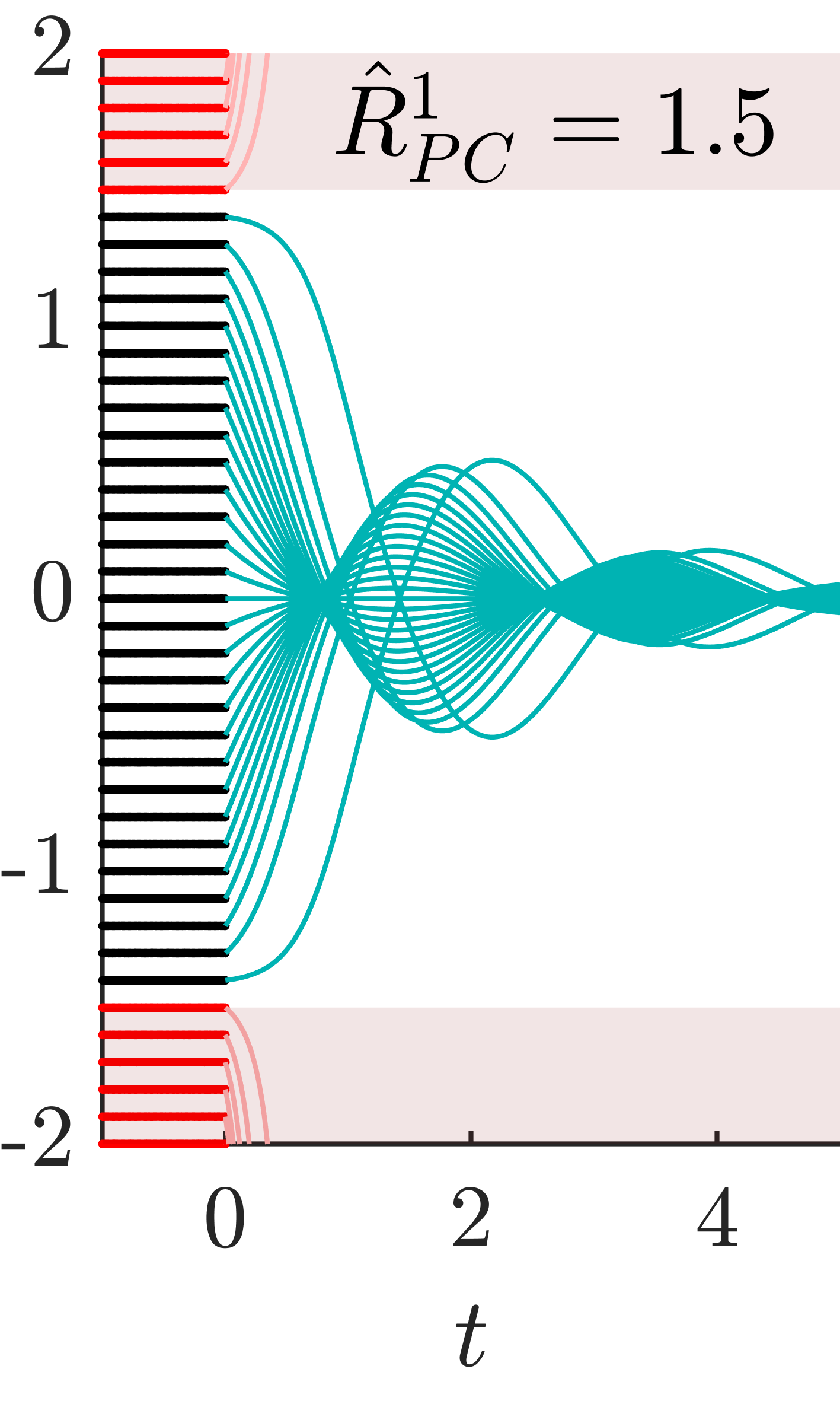}
}
\hfill
 \subfloat[]{
\includegraphics[]{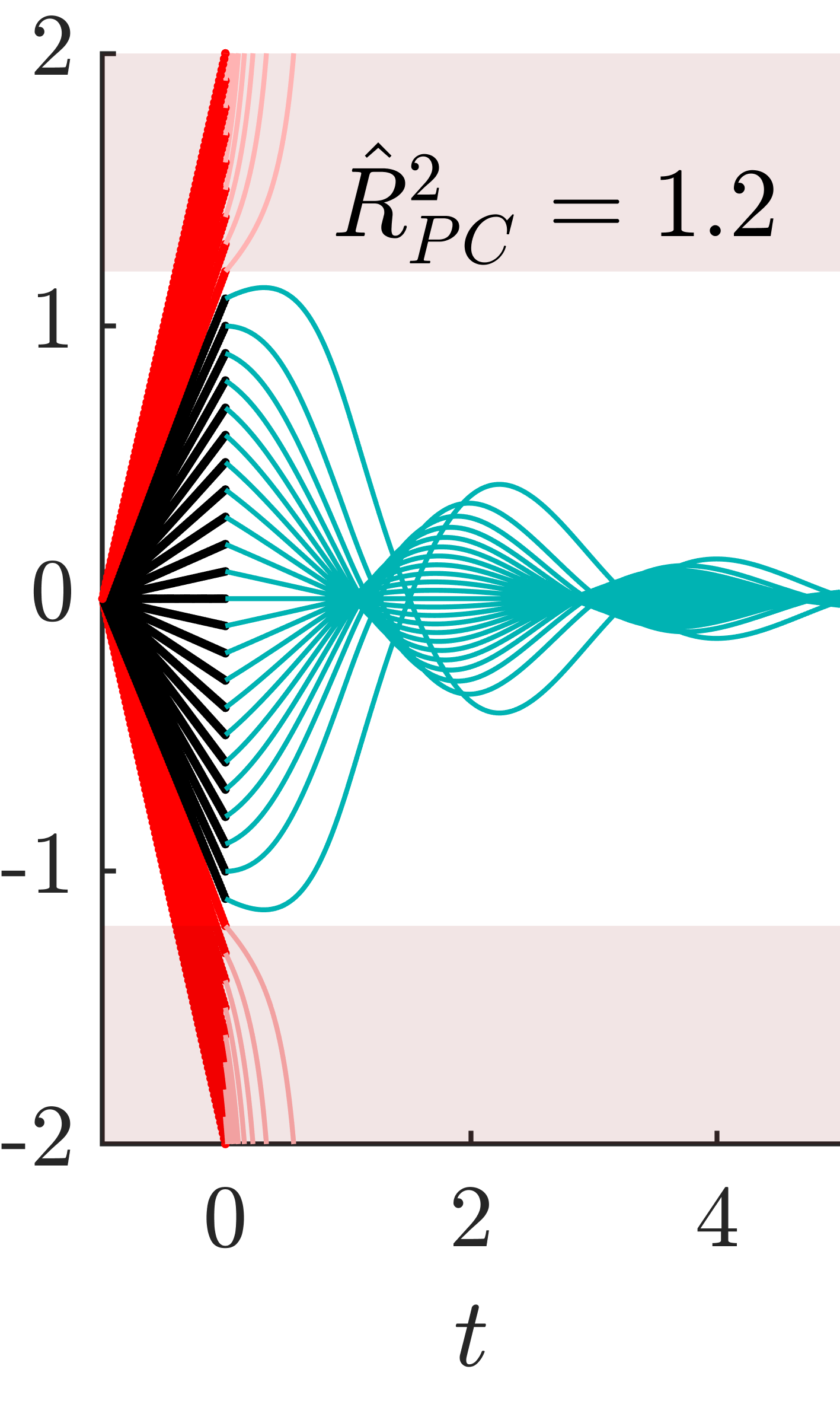}
}
\hfill
 \subfloat[\label{fig:scalar_step1}]{
\includegraphics[]{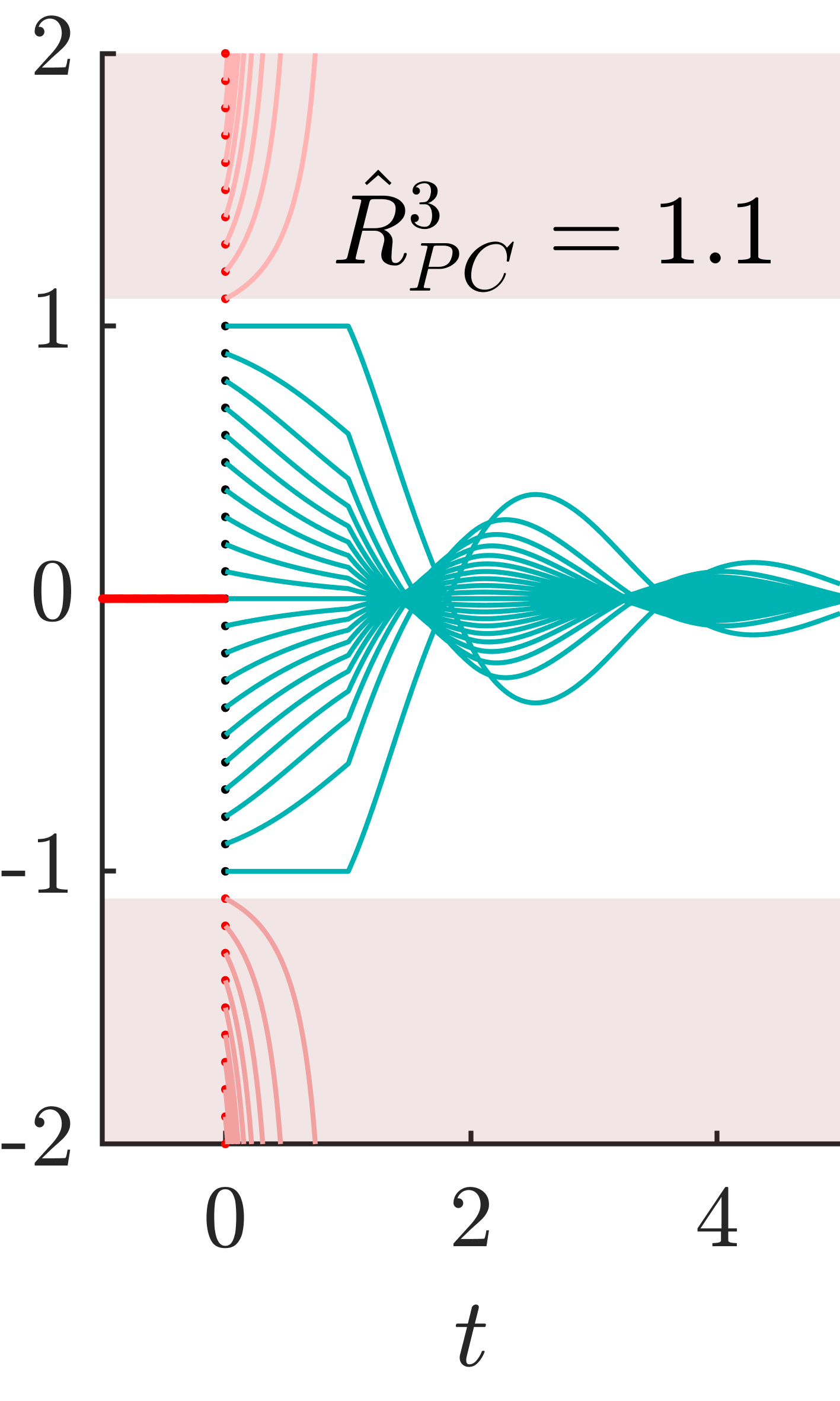}
}
\hfill
 \subfloat[\label{fig:scalar_lin_dec1}]{
\includegraphics[]{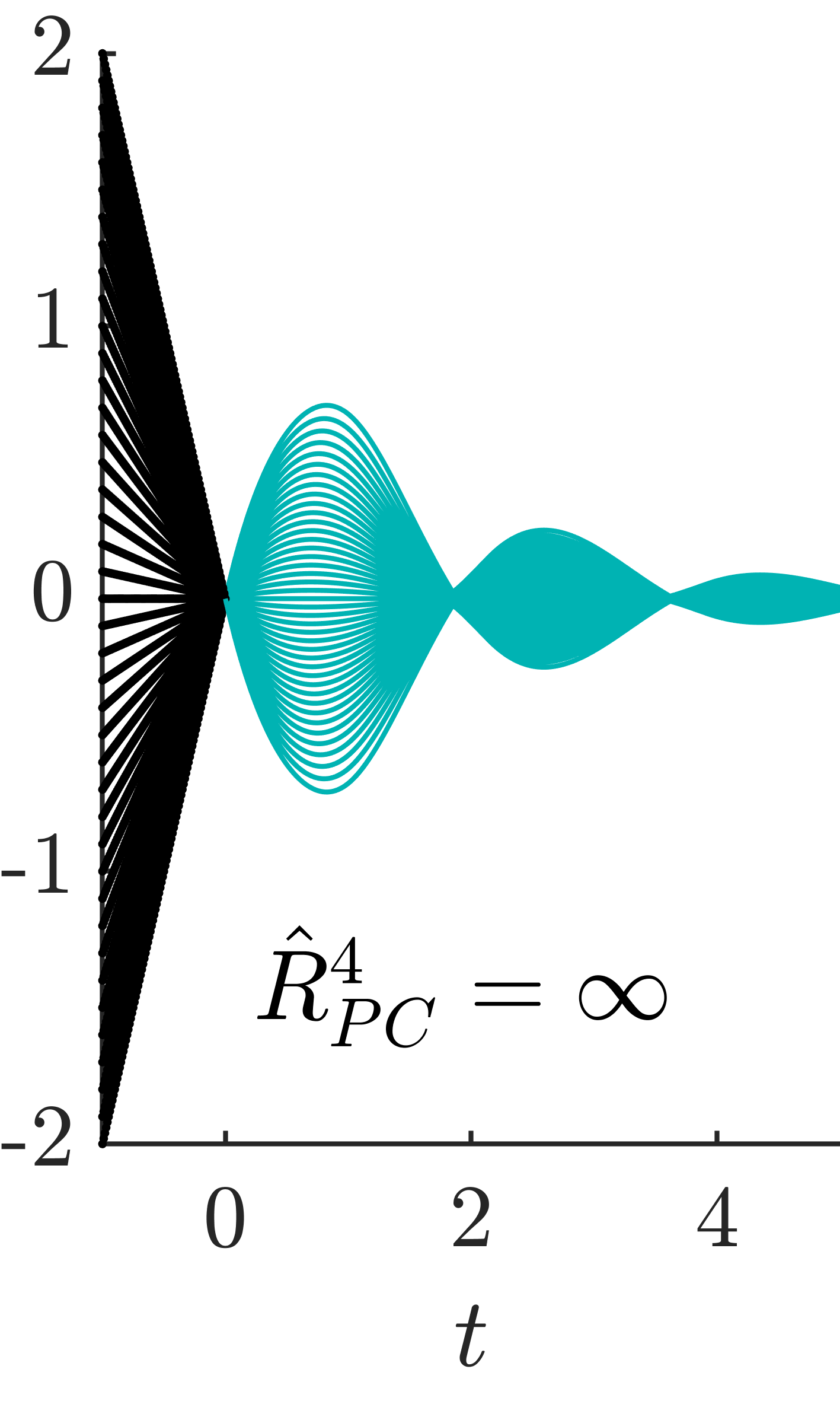}
%\hspace{-3em}\input{fig_E_Legend.pdf_tex}
}
\hspace{-0.75cm}
\begin{minipage}[b]{2.5cm}
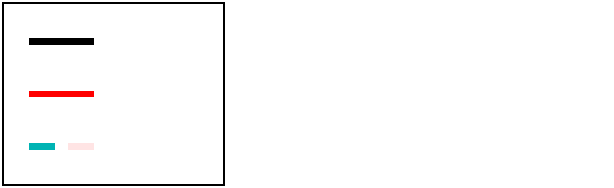 %\input{fig_E_Legend.pdf_tex}
\vspace{10em}
\end{minipage}
%\hspace{-3em}
 %\subfloat[]{
%\input{fig_E_Legend.pdf_tex}
%\vspace{11em}
%}
\\
 %\subfloat[]{
%$\tau=5:$
%\vspace{11em}
%}
\begin{minipage}[b]{1.5cm}
$\tau=5:$
\vspace{9.5em}
\end{minipage}
 \subfloat[]{
\includegraphics[]{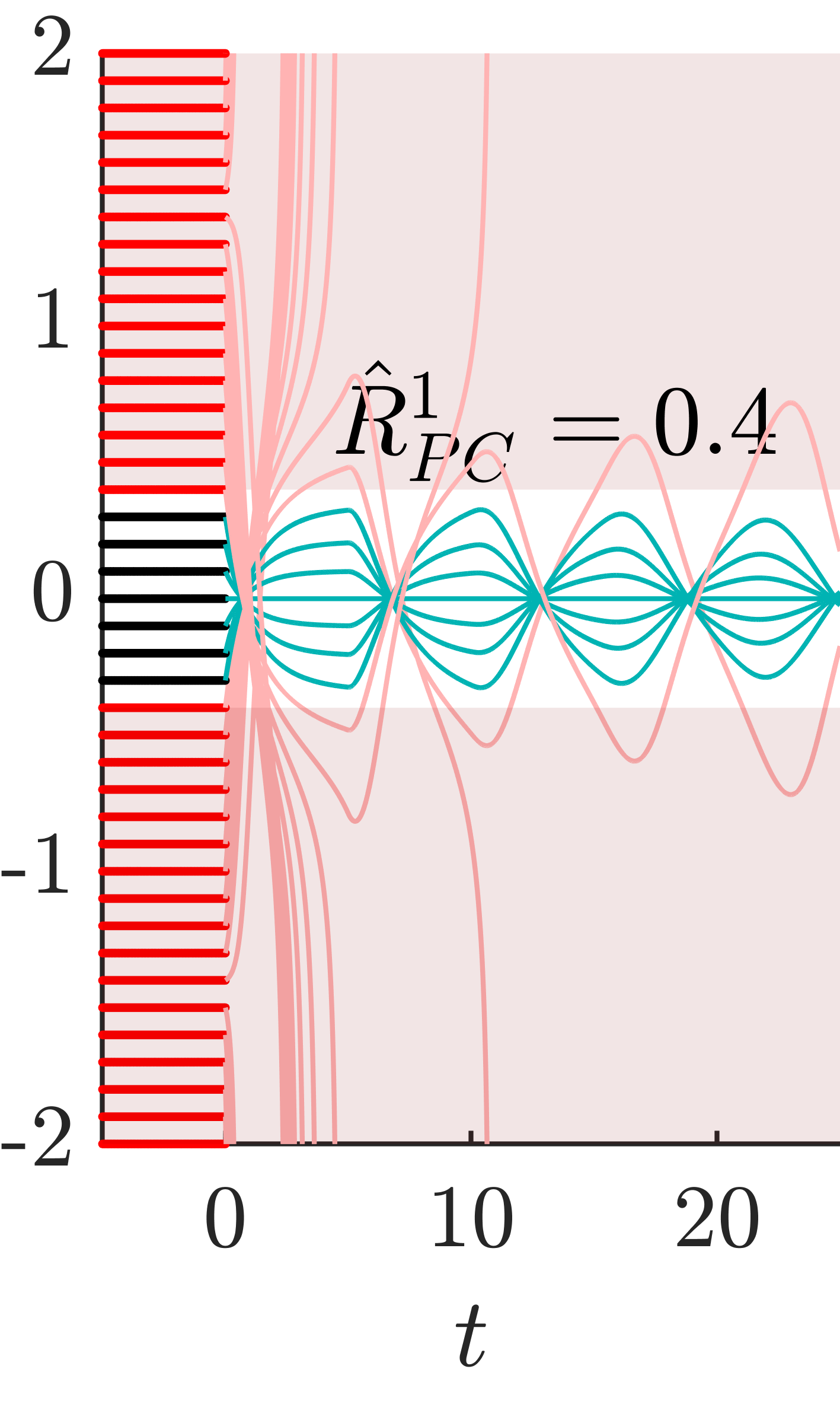}
}
\hfill
 \subfloat[]{
\includegraphics[]{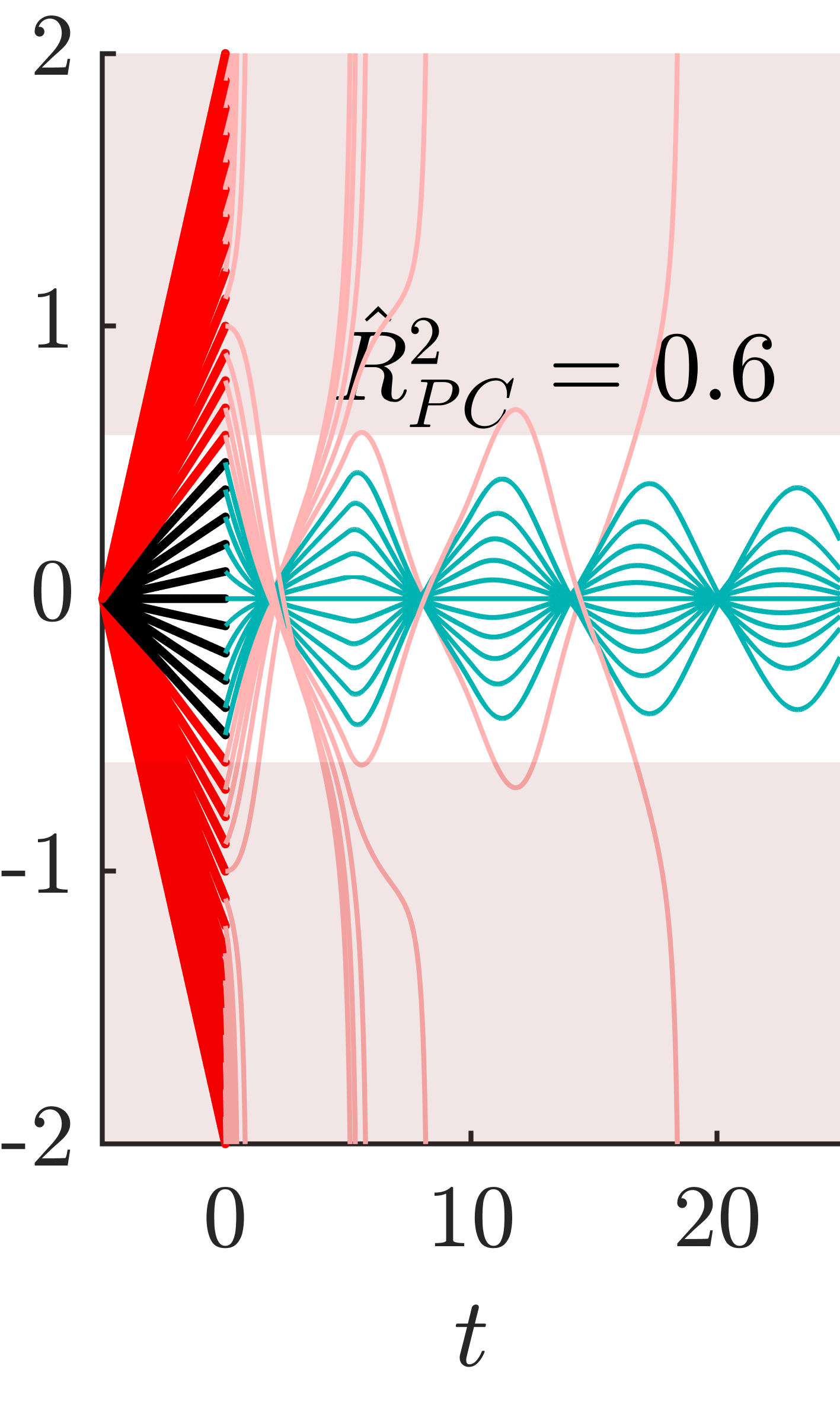}
}
\hfill
 \subfloat[\label{fig:scalar_step5}]{
\includegraphics[]{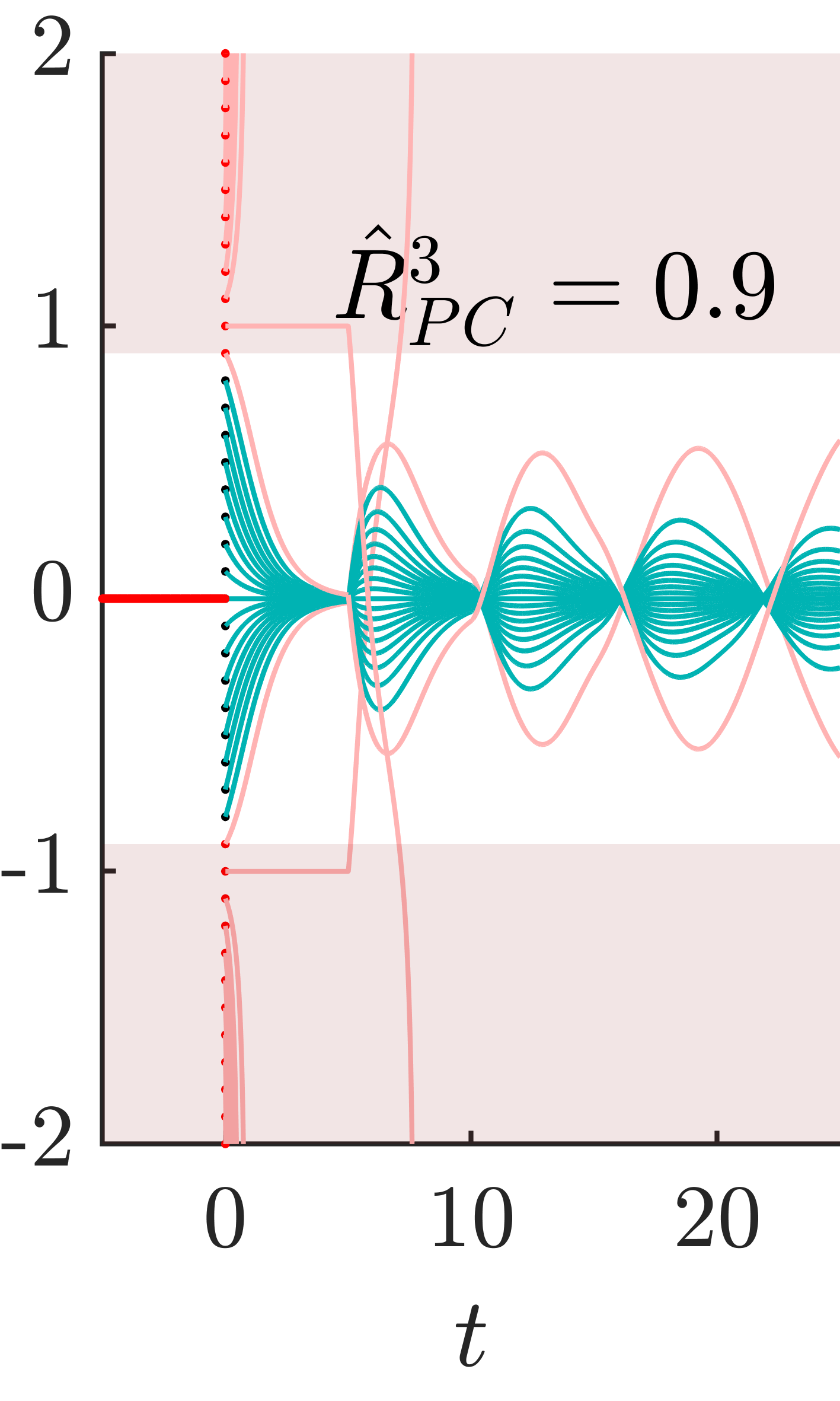}
}
\hfill
 \subfloat[\label{fig:scalar_lin_dec5}]{
\includegraphics[]{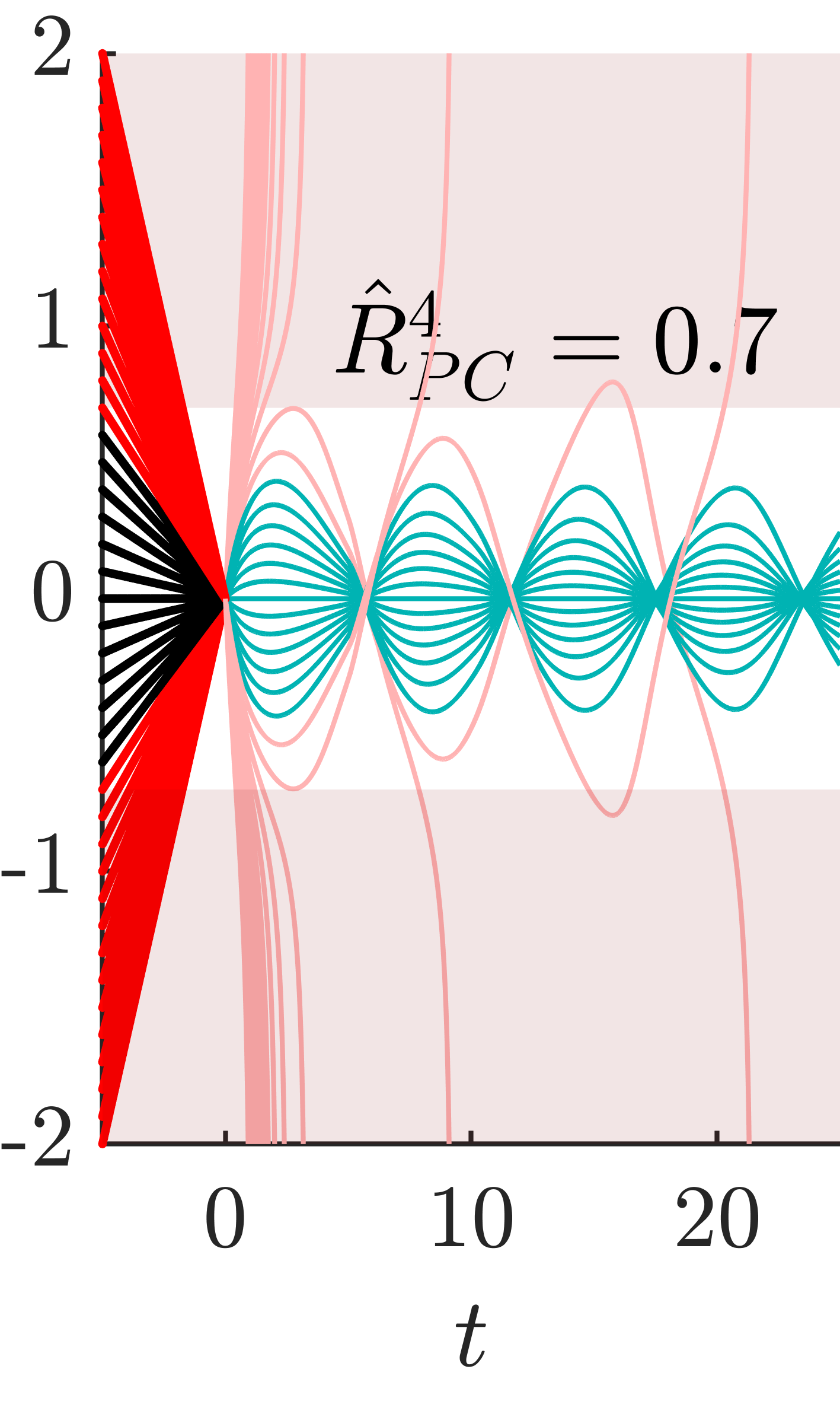}
}
\begin{minipage}[b]{1.75cm}
~
\end{minipage}
%!!!!!!!!!!!!!!!!!!!!!!!!!!!!!!!!!!!!!!!!!!!!!!!!!!!!!!!!!!!!!!!!!!!!!
\else
%!!!!!!!!!!!!!!!!!!!!!!!!!!!!!!!!!!!!!!!!!!!!!!!!!!!!!!!!!!!!!!!!!!!!!
\centering
%\begin{minipage}[b]{1.5cm}
%$\tau=1:$
%\vspace{9.5em}
%\end{minipage}
 \subfloat[]{
\includegraphics[width=0.23\textwidth]{fig_E_tau_1_i_1.png}
}
\hfill
 \subfloat[]{
\includegraphics[width=0.23\textwidth]{fig_E_tau_1_i_2.png}
}
\hfill
 \subfloat[\label{fig:scalar_step1}]{
\includegraphics[width=0.23\textwidth]{fig_E_tau_1_i_3.png}
}
\hfill
 \subfloat[\label{fig:scalar_lin_dec1}]{
\includegraphics[width=0.23\textwidth]{fig_E_tau_1_i_4.png}
%\hspace{-3em}\input{fig_E_Legend.pdf_tex}
}
\hspace{-2.3cm}
\begin{minipage}[b]{2.2cm}
\small
\input{fig_E_Legend.pdf_tex} %\input{fig_E_Legend.pdf_tex}
\vspace{10em}
\end{minipage}
\\
%\begin{minipage}[b]{1.5cm}
%$\tau=5:$
%\vspace{9.5em}
%\end{minipage}
 \subfloat[]{
\includegraphics[width=0.23\textwidth]{fig_E_tau_5_i_1.png}
}
\hfill
 \subfloat[]{
\includegraphics[width=0.23\textwidth]{fig_E_tau_5_i_2.png}
}
\hfill
 \subfloat[\label{fig:scalar_step5}]{
\includegraphics[width=0.23\textwidth]{fig_E_tau_5_i_3.png}
}
\hfill
 \subfloat[\label{fig:scalar_lin_dec5}]{
\includegraphics[width=0.23\textwidth]{fig_E_tau_5_i_4.png}
}
%\begin{minipage}[b]{1.75cm}
%~
%\end{minipage}
%!!!!!!!!!!!!!!!!!!!!!!!!!!!!!!!!!!!!!!!!!!!!!!!!!!!!!!!!!!!
\fi
\caption{\label{fig:scalarExample} Solutions $x(t;\phi)$ of (\ref{eq:scalarExample}) 
for constant \rev{$\phi(\theta)\equiv p$, $p\in \mathbb R$}, linear / increasing \rev{$\phi(\theta)=p (\theta+\tau)/ \tau$}, jump \rev{$\phi(\theta)=p \chi_{\{0\}}(\theta)$} and linear / decreasing \rev{$\phi(\theta)=p \theta/\tau$} initial functions ${\phi(\theta), \theta\in[-\tau,0]}$ and two different time delay values ((a)-(d): $\tau=1$, (e)-(h): $\tau=5$)}
\end{figure*}
%%%%%%%%%%%%%%%%%%%%%%%%%%
Four families of initial functions in $PC([-\tau,0],\mathbb R)$ are considered.
 We are interested in the radius of attraction $R_{PC}$, i.e.\ the maximum uniform norm $\| \phi \|_C$  below which all solutions still converge to the zero equilibrium. Clearly, the tested set of initial functions yields $R_{PC}\leq\revi{\hat R_{PC}(\Phi)}=\min_i \hat R_{PC}^i =1.1$ if $\tau=1$ and $\revi{\hat R_{PC}(\Phi)}=0.4$ if $\tau=5$.
In the case of a delay value $\tau=1$ the family of step functions (Fig. \ref{fig:scalar_step1}) shows the smallest attracted set and the family of decreasing linear functions with $\phi(0)=0$ (Fig. \ref{fig:scalar_lin_dec1}) is converging for all tested functions. However,  in the case of $\tau=5 $ the step functions (Fig. \ref{fig:scalar_step5}) even show the largest attracted set, whereas the decreasing linear functions (Fig. \ref{fig:scalar_lin_dec5}) already diverge for smaller $\|\phi\|_C$ values. Hence, the class of functions \rev{that} leads to the best estimation is not  even constant for one type of RFDE under parameter variations. As a consequence, there is no obvious relation between the class of initial functions and the conservativeness of estimations on $R_X$. 
\end{example}
%%%%%%%%%%%%%%%%%%%%%%%%%%%%%%%%%%%%%%%%%%%%%%%%%%%%%%%%%%%%%%%%%%%%%%%%%%%%%%%%%%%%
\subsection{Generalizations of Domain and Radius of Attraction}
The radius of attraction $R_X$ is dictated by the smallest initial function \rev{that} leads to a non-zero-convergent solution -  no matter how relevant this initial function might be. As schematically shown by the intruding blue line in Figure \ref{fig:Measures}, the most critical element in the complement $\mathscr D_X^c$ is possibly only a small exception within the closure of the domain of attraction. 
\\
\\
 For ODEs, \rev{there are} generalizations of the domain of attraction. \rev{Such generalizations are either} based on topological concepts, like  quasi-stability \rev{domains} \cite{Zaborszky.1988,Chiang.1996,Chiang.2011}, or based on measure theoretic concepts, like almost global stability \cite{Rantzer.2001} (see Appendix). 
\rev{These generalizations form supersets of the exact domain of attraction. They ignore that some small sets in the interior actually do not lead to the desired zero-convergence, and\rev{,} hence\rev{,} do not belong to the exact domain of attraction. As a consequence, inner estimations of the generalized domain of attraction might be considerably larger than inner estimations of the exact domain of attraction - for the price of marginal exceptions \rev{that} might be overseen.}
\\
\\
In this light, we propose definitions of generalized radii of attraction $R_X^{gen}$ in time\rev{-}delay systems in the \rev{A}ppendix. It can be expected that upper bounds $\hat R_X$ \rev{that} are based on arbitrarily selected initial functions, give only an estimation of the such introduced $R_X^{gen}$ (Figure \ref{fig:Measures}).
\rev{However, $R_X^{gen}$ might even be the actual number of interest: if a statement about generic convergence to the attractor suffices, $R_X$ would be unnecessarily restrictive}.

\begin{remark}[No Lebesgue Measure in Infinite Dimensional Spaces]\label{rem:noLebesgueMeasure}
The notion of a property to be valid for almost every (a.e.) point in the state space or a subdomain is convenient in the context of ordinary differential equations. If not otherwise stated, exclusions of Lebesgue measure zero are meant. However, there is no analog of the Lebesgue measure on infinite dimensional spaces (Sullivan \cite{Sullivan.2015}, Theorem 2.38). As a consequence, there is no universally accepted notion of ``almost every function'' \cite{Ott.2005}. 
\end{remark}
%%%%%%%%%%%%%%%%%%%%%%%%%%%%%%%%%%%%%%%%%%%%%%%%%%%%%%%%%%%%%%%%%%%%%%%%%%%%%%%%%%%%
\section{On the Selection of the State Space} \label{sec:SelectionOfTheBanachSpace} 
There are various normed spaces - predominantly Banach spaces - \rev{that} come into question as state space for RFDEs. In literature, the space of continuous functions $C([-\tau,0],$ $\mathbb R^n)$ endowed with the uniform norm is mostly used \cite{Hale.1993,Diekmann.1995,Bellman.1963,Halanay.1966,LaSalle.1976}. Also frequent \cite{Krasovskii.1964,Delfour.1972,Bernier.1978,Curtain.1995} are spaces $M^p=  \mathbb R^n \times L^p([-\tau,0],\mathbb R^n)  $ with  $1\leq p \leq\infty$, especially $M^2$. A $L^p$ space alone would not be sufficient, since a RFDE requires the point value $x(0)=\phi(0)\in \mathbb R^n$ at $t=0$ to start with, which is a set of Lebesgue measure zero. In $M^p$ both, the initial function $\phi\in L^p$ and the initial value $\phi(0)\in \mathbb R^n$, are given separately.  Space selections in literature seem at most to be aligned to the mathematical theory. However, due to non-equivalence of norms in infinite dimensional spaces, the selection of the normed space has far reaching consequences and should be well considered. 
\subsection{Consequences of the Selection}\label{sec:consequecesSpace}
According to (\ref{eq:attractive}), there is a non-zero radius of attraction $R_X\neq 0$ in any normed space\rev{,} in which the equilibrium under consideration can be proven to be attractive.
However, the interpretation as well as the size of the radius of attraction depend heavily on the chosen norm.   For instance,
the radius of attraction might be a measure for maximum allowed pointwise perturbations ($\|\cdot\|_C$) or address derivations in energy ($\|\cdot\|_{L^2}$ as part of $\|\cdot\|_{M^2}$). Thus, the technical or physical application is of interest.
\\
\\
The space also determines which initial functions come into question. Consider delayed input in a control law. It requires a certain initialization until first data is received. As soon as the signal becomes available, this initialization is replaced by the actual value, such that a jump discontinuity occurs unless the values are equal by chance. If the start-up process or a communication interruption is only one of many possible scenarios, not only the finite dimensional case of initial jumps, but an infinite dimensional function space, which allows such discontinuous functions, has to be taken into account.
This example shows that the Banach space of continuous functions $C([-\tau,0],\mathbb R^n)$ is frequently not sufficient. However, the Banach space of piecewise continuous functions $PC([-\tau,0],\mathbb R^n)$ equipped with the uniform norm is a valid alternative for RFDEs \cite{Krasovskii.1963,Kharitonov.2013} and does not alter the meaning of the radius of attraction.
Obviously, other spaces like $M^2$ also allow discontinuous functions, but the interpretation of the radius of attraction is completely different.
\begin{remark}[Invariance under Time Scale Transformations]
Time scale transformations are convenient to get dimensionless parameter values in appropriate scales and to lower the number of parameters. However, by affecting the time delay value, time scale transformations do also change the domain of the initial function. What are the consequences for the radius of attraction? For instance, a delay normalizing transformation %\revi{of (\ref{eq:RFDE}) with (\ref{eq:RFDE_singleDelay})} 
is achieved by $t=\tau \hat t$, $x(\tau \hat t)=\hat x (\hat t)$ and $\theta=\tau \hat \theta$, $\phi(\tau \hat \theta)=\hat \phi (\hat \theta)$, respectively. Hence, $\dot x = f(x(t),x(t-\tau)), x_0=\phi$ becomes 
\begin{align}
\frac{\textrm d \hat x}{\textrm d \hat t}(\hat t)&=\tau f(\hat x(\hat t),\hat x(\hat t-1)), &&\hat t>0 \\
\hat x(\hat \theta)&=\hat \phi(\hat \theta):=\phi(\tau \hat \theta), &&\hat \theta\in[-1,0]. \nonumber
\end{align}
The radius of attraction with respect to the transformed system refers to $\hat\phi(\hat \theta)$. 
Obviously, by 
\begin{align}
\|\hat \phi\|_{C[-1,0]} = \sup_{\hat \theta\in[-1,0]}\phi(\tau \hat \theta) = \| \phi\|_{C[-\tau,0]}
\end{align}
the uniform norm of the initial function and thus the radius of attraction in $C$ or $PC$ is invariant under such transformations.
However, for instance the norm in $M^2=\mathbb R^n \times L^2$ is not since
\begin{align}
\| \hat \phi\|_{M^2[-1,0]}
&=\|\hat \phi(0)\|_2 + \left(\int_{-1}^0 \|\phi(\tau \hat\theta)\|_2^\rev{2}  \textrm d \hat \theta\right)^{\frac 1 2}
\nonumber\\
&=\| \phi(0)\|_2+\sqrt{\frac 1 \tau} \left(\int_{-\tau}^0 \|\phi(\theta)\|_2^\rev{2}  \textrm d  \theta\right)^{\frac 1 2}
\nonumber \\
&\neq \|\phi\|_{M^2[-\tau,0]}.
\end{align}
As a consequence, the radius of attraction in $M^2$ of the transformed system has no meaning for the original RFDE.
\end{remark}
\subsection{Quotient Space to Incorporate Differently Delayed States}
In physical or technical problems it is common that some state components occur in a delayed form, but others contribute only by their instantaneous values (see \rev{e.g. \cite{Dombovari.2008, Scholl.2019b, Minorsky.1948} or} Example \ref{ex:swingEq} below).
Hence, the history of some state variables $x^{\textrm{\scriptsize II}}(t)\in \mathbb R^{q}$, $q<n$ in $x(t)=[x^{\textrm{\scriptsize I}} (t)^\top, x^{\textrm{\scriptsize II}}(t)^\top]^\top\in \mathbb R^n$ does not influence the system dynamics at all, i.e.\
\begin{align}
\begin{bmatrix} 
\dot x^{\textrm{\scriptsize I}}(t) \\ \dot x^{\textrm{\scriptsize II}}(t)
\end{bmatrix} 
&= f(x^{\textrm{\scriptsize I}}(t), x^{\textrm{\scriptsize II}}(t), x^{\textrm{\scriptsize I}}(t-\tau) ), && t>0 \\
\begin{bmatrix} x^{\textrm{\scriptsize I}} (\theta) \\
x^{\textrm{\scriptsize II}}(\theta)\end{bmatrix} 
&= \begin{bmatrix}\phi^{\textrm{\scriptsize I}} (\theta)\\ \phi^{\textrm{\scriptsize II}}(\theta)\end{bmatrix} , && \theta\in[-\tau,0]. \nonumber
\end{align}
Usually, there is no distinction in the domains of the initial function components, such that $\textrm{dom}(\phi^{\textrm{\scriptsize II}})=\textrm{dom}(\phi^{\textrm{\scriptsize I}})=[-\tau,0]$. Thereby, it is ignored that some of the required initial data $[\phi^{\textrm{\scriptsize I}{\top}}, \phi^{\textrm{\scriptsize II}{\top}}]^\top=\phi\in X([-\tau,0] , \mathbb R^n)$ remains unused (cmp.\ \cite{Lee.1982}). Thus, the Banach space norm $\|\phi\|_X$ depends on irrelevant values of $\phi^{\textrm{\scriptsize II}}$ for $\theta<0$. Consider $\|\phi^{\textrm{\scriptsize II}}\|_X>\|\phi^{\textrm{\scriptsize II}}(0)\|$. Obviously, the radius of attraction as a norm criterion on $\|\phi\|_X$ would be unnecessarily restrictive. 
\\
\\
This situation motivates the use of 
\begin{align}
\|\phi\|_{Q}:=
\left \|
\begin{array} {c}
\| \phi^{\textrm{\scriptsize I}}\|_X\\
 \| \phi^{\textrm{\scriptsize II}}(0) \| 
\end{array} 
\right\|
\label{eq:QNorm}
\end{align}
with a certain choice of the outer norm  (which is not that important by equivalence of norms in finite dimensions). 
However, $\|\cdot\|_{Q}$ is only a seminorm in $X$. Therefore, we propose to introduce the quotient space
\begin{align}
Q:=X / \{\phi \in X: \phi^{\textrm{\scriptsize I}}(\theta)\equiv 0_{{n-q}}, \phi^{\textrm{\scriptsize II}}(0)=0_{q} \}
\label{eq:QuotientSpace}
\end{align}
endowed with the norm  $\|\cdot\|_{Q}$.
It regards functions to be in the same equivalence class if they differ in values $\phi^{\textrm{\scriptsize II}}(\theta)$ for $\theta<0$.
Hence, these values are fully ignored in the sufficient criterion for zero-convergence  $\|\phi\|_{Q}<R_Q$.
%\begin{remark}[Differently Delayed States in Other Contexts]
%Coexisting state variables with and without delay or different maximum delays are rarely addressed in literature. Gu \cite{Gu.2010} discusses multiple delay channels with focus on stability. In the context of controller synthesis, Delfour, Lee, and Manitius introduce the so-called F-reduction  \cite{Delfour.1978,Lee.1982} to overcome the problem.
%\end{remark}
%%%%%%%%%%%%%%%%%%%%%%%%%%%%%%%%%%%%%%%%%%%%%%%%%%%%%%%%%
\section{Estimations by Time\rev{-}Forward Simulations}\label{sec:SimApproach}
The state space $X$ is infinite dimensional.
However, a simulative approach means to select certain initial functions. 
Thereby, Example \ref{ex:scalarExample} demonstrates that an obvious relation between the class of initial functions and the conservativeness of the overestimation of %the radius of attraction 
$\hat R_X$ does not exist.  Hence, the question arises, which initial functions should be taken into account.
\\
\\
The physical or technical context might already motivate certain initial functions.
\\
{\itshape (i)} Predefined parameterized initial functions might be of particular importance, e.g., jump functions for delayed controller input (cf. Section \ref{sec:consequecesSpace}). \\
{\itshape (ii) }In the case of a switched system, the initial function should correspond to the dynamics of the previous system definition. \\
{\itshape (iii) }There is also the proposal to use solution segments built from the solutions of the system without delayed terms  \cite{Daza.2017}. %nur aufgezählt
\\
\\
If application-adapted functions occur exclusively, only the corresponding \rev{subset} of $X$ is of interest. Consequently, the intersection of the domain of attraction with this \rev{subset} suffices. Otherwise, a large span of different families of initial functions should be tested. 
In the following section various possibilities of initial functions are classified.  
The norm of the smallest found initial function \rev{that} does not belong to the domain of attraction is of interest. According to Definition \ref{def:hat_R_X}, it gives the desired overestimation $\hat R_X\geq R_X$ of the radius of attraction.
\\
\\
\rev{
Results in the present paper are generated by the Matlab solver \texttt{dde23}, which is based on a Runge-Kutta triple \cite{Shampine.2001}.  \revi{Stiff problems might require alternative approaches \cite{Agrawal.2004}. For further details on numerical methods we refer to \cite{Bellen.2003}}.
No matter which solver is chosen, a tuning of the discretization step size or error bounds is inevitable. Furthermore, the criterion according to which a numerical solution is assumed to approximate a zero-convergent solution has to be selected. Consider a numerical solution over the finite time span $t\in[-\tau,t_\revi{\textrm{num}}]$ for a given initial function $\phi$. A convergence decision must be based on the approximation of the last state $x_{t_\revi{\textrm{num}}}(\cdot;\phi)$, i.e.\ on the numerical data for $t\in [t_\revi{\textrm{num}}-\tau, t_\revi{\textrm{num}}]$. 
\revi{Let $r(t_{\textrm{num}};\phi):=\|\hat x_{t_\revi{\textrm{num}}}(\cdot;\phi)\|_X $ be the respective norm of the interpolated numerical values.} 
%Let $\hat z^i \in \mathbb R^{N} $, $i\in \{1,\ldots,n\}$ denote an equidistant discretization of the $i$-th state component for this last delay-spanning  segment of the numerical solution. 
%Let $\|\cdot \|^i$ be an appropriate norm for this $i$-th component, e.g. $\|\hat z ^i\|^i = \frac{1}{N}\| z ^i\|_1=\frac{1}{N}\sum_{j=1}^N  \vert z^i_j\vert$. 
The numerical domain of attraction is
\begin{align} \label{eq:hat_D_X_numerical}
%\hat{\mathcal D}_X=\{\phi \in X:  \frac{1}{n} \sum_{i=1}^n \| \hat z^i \|^i<\delta_{\textrm{num}} \} 
\hat{\mathcal D}_{X}\stackrel{\textrm{def}}=\{\phi \in X:  \revi{r (t_\textrm{num}; \phi)}<\delta_{\textrm{num}} \}, 
\end{align} 
\revi{cmp.\ (\ref{eq:attractive}), }with $\delta_{\textrm{num}}>0$ suitably chosen. \revi{If a lower bound on the radius of attraction $\check R_X\leq R_X$ is available $\delta_{\textrm{num}}$ should be based on this result. Furthermore, $\delta_{\textrm{num}}$ can be used as a termination criterion in the simulation to lower the numerical effort.}
%In the following, with a slight abuse of notation, we identify this numerical domain of attraction with $\mathcal D_X$.
In the following, with a slight abuse of notation, we identify the numerical domain of attraction $\hat {\mathcal D}_{X}$ with $\mathcal D_X$.
}
\rev{\begin{remark}[Numerical Errors]
All results on the domain of attraction that are based on time simulations \cite{Leng.2016, Yan.2019, Schafer.2015,Shang.2009, Aguirregabiria.1987, Daza.2017,Losson.1993,Taylor.2007} rely on the distinction between convergence and divergence from numerical results. However, this asymptotic behavior $t\to \infty$ is usually not covered by any numerical bounds, and thus the classification cannot be considered as proven. Concerning the radius of attraction, trajectories wrongly classified as convergent will not contribute to the upper bound estimation, but trajectories wrongly classified as non-zero-convergent might give a too small upper bound. While from a mathematical point of view this is not satisfactory, it should be noted that if the numerical errors (e.g.\ numerical damping) are able to generate non-convergence for a given initial function, then comparable errors in system parameters (e.g.\ physical damping), equation structure, or external perturbations will do so a fortiori. 
\end{remark}}
\subsection{Primary Initial Functions}\label{sec:primary}
We use the term primary initial functions to describe simple functions $\phi^k(\, \cdot \, ;p)$, $k\in\{1,\ldots, \rev{K}\}$ like constant, jump, polynomial, or trigonometric functions, which are parameterized in $p\in \mathbb R^{m}$. Denote the set of all elements in a family $\{\phi^k\}_p$ of initial functions of type $k$
 shortly as
\begin{align}
\Phi^k\stackrel{\textrm{def}}=\left\{ \phi^k(\,\cdot \,;p)\in X : p\in \mathbb R^{m} \right \}
\label{eq:initFuncFamily}
\end{align}  
and the set of all considered functions as 
\begin{align}
\Phi\stackrel{\textrm{def}}=\bigcup_k \Phi^k.
\label{eq:allPhi}
\end{align}
There are few considerations of domains of attraction in literature. However, the existing ones are mostly based on scalar linear  or  sinusoidal  initial functions dependent on one  \cite{Losson.1993} or two  \cite{Daza.2017,Taylor.2007} parameters $(p_1,p_2) \in \mathbb R^2$. 
In the latter case a graphical representation of the $(p_1,p_2)$ plane is common. 
By marking those parameter combinations \rev{that} result in an attractor-convergent solution, the finite dimensional intersection $\Phi^k \cap \mathcal D_X$ of the domain of attraction with the chosen family of primary initial functions becomes visible (black pixels in Figure \ref{fig:ROA_pixel}). 
\begin{figure} 
\centering
 \subfloat[\label{fig:DOA_pixel_const} $\phi^1_2(\theta)\equiv p_2$]{
\includegraphics[width=0.47\linewidth]{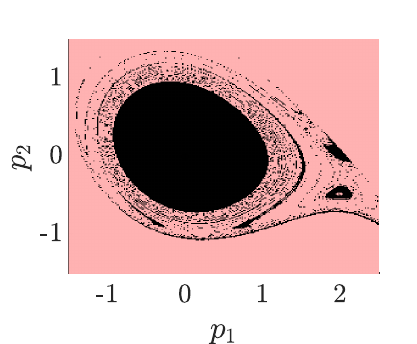}
}
\hfill
 \subfloat[$\phi^2_2(\theta)=p_2 \,\chi_{\{0\}}(\theta)$]{
\includegraphics[width=0.47\linewidth]{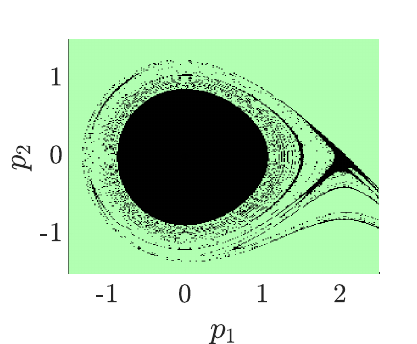}
}
\\
 \subfloat[$\phi^3_2(\theta)=p_2\cos\left(\frac{4\pi}{\tau} \theta\right)$]{
\includegraphics[width=0.47\linewidth]{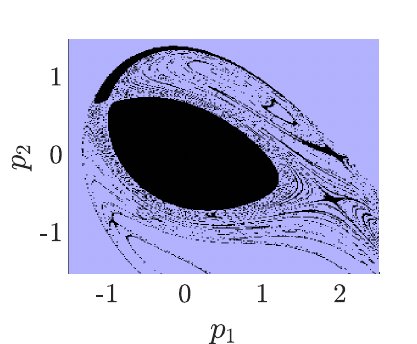}
}
\hfill
 \subfloat[$\phi^4_2(\theta)=p_2\sin\left(\frac{4\pi}{\tau} \theta\right)$]{
\includegraphics[width=0.47\linewidth]{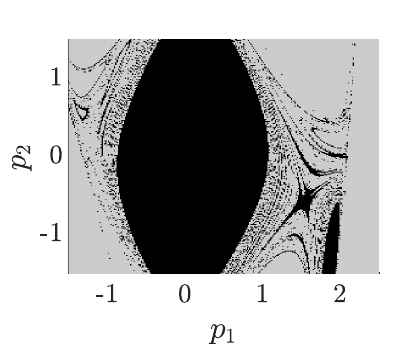}
}
\caption{\label{fig:ROA_pixel} Intersection of the domain of attraction with certain families of primary initial functions $\Phi^k \cap \mathcal D_X$. Black pixels indicate parameter combinations \rev{that numerically} lead to zero-convergent solutions \rev{in (\ref{eq:delayedSwingEq}), (\ref{eq:param_swing_eq})}, i.e.\ \rev{for black $p_1$,$p_2$ combinations }$\phi^k(\,\cdot\,;[p_1,p_2]^\top)\in \rev{\hat{\mathcal D_X}}$ \rev{holds}.% \mathcal D_X$. 
}
\end{figure}
\\
\\
These plots are by far more computationally expensive than phase plots for ordinary differential equations. By the semigroup property, each point $x(t)$ of an ODE solution can be seen as new initial point and does not need to be checked again for zero-convergence. However, the states $x_t$ in a delay system are whole solution segments \rev{that} are probably not in the chosen subset of initial functions $x_t\not\in \Phi^k$. 
In view of avoiding unnecessary computational burden, the examination of the whole parameter space is not undertaken:
to derive estimations 
%$\hat R_X$  
\revi{$\hat R_X(\Phi)$} 
for the radius of attraction, it is sufficient to increase the parameters until zero-convergence fails. This can be done in a \rev{circular or} star-like scheme, 
in which the parameter modulus %along the rays 
is increased (Figure \ref{fig:p1_p2_plane} below). 
 Of course, symmetries should be taken into account to reduce the computational effort further. In addition, the problem is well suited for parallel computing. 
%%%%%%%%%%%%%%%%%%%%%%%%%%%%%
\begin{remark}[Fractal Domain of Attraction Boundaries]\label{rem:fractal}
The boundary of a domain of attraction is possibly fractal \cite{Aguirregabiria.1987,Daza.2017,Losson.1993,Taylor.2007}. Fractality results in a high sensitivity of simulative results on parameter variations in the initial function. However, it does not influence the validity of gained upper bounds on the radius of attraction. 
\end{remark}
\noindent It should be noted that physics might give restrictions on how the state components are related to each other.
\begin{remark}[Higher Order Derivatives]\label{rem:scalarHigherOrder}
Assume that the state space representation $\dot x(t)=F(x_t)$ with $x(t)\in \mathbb R^n$ stems from a system description in a scalar variable $y(t)\in \mathbb R$ with higher order derivatives 
\begin{align}
y^{(n)}(t)+g\big(&y^{(n-1)}(t),\ldots,\dot y(t),y(t),
\nonumber \\
&y^{(n-1)}(t-\tau),\ldots,\dot y(t-\tau),y(t-\tau)\big) = 0.
\end{align}
Then, the transformation to $x$, e.g., $x=[y,\dot y, \ldots, y^{(n-1)}]^\top$,  goes along with restrictions to the space of initial functions. Hence,  the $n$-dimensional vector function $\phi(\,\cdot\,;p)=x_0(\,\cdot\,;p)$ cannot freely be chosen. Instead, only a scalar initial function $y_0(\,\cdot\,;p)$ for $y(t)$ must be specified. For instance, with the transformation from above, this results in
\begin{align}
\phi(\theta;p)=[y_0(\theta;p), \dots, y_0^{(n-1)}(\theta;p)]^\top.
\end{align}
\end{remark}
\subsection{Secondary Initial Functions} \label{Sec:SecondaryInitialFunctions}
We use the term {\itshape secondary initial functions} for functions \rev{that} arise a posteriori from previous simulation results. 
%Let $\mathscr T(t) \colon X \to X$, $\phi \mapsto \mathscr T(t)\phi:=x_t(\cdot;\phi)$ denote the solution operator. As described above, the semigroup property $\mathscr T(s+t)=\mathscr T(t) \mathscr T(s)$, $t,s\geq 0$,
As described in Section \ref{sec:primary},  the semigroup property $x_{s+t}(\cdot;\phi)=x_t(\cdot; x_s(\cdot;\phi))$ does not help much with respect to considerations of $\Phi^k \cap \mathcal D_X$, which are addressed in graphical representations like Figure \ref{fig:ROA_pixel}. Nevertheless, it does help for considerations of the whole domain of attraction $\mathcal D_X$, which indeed is addressed by the radius of attraction. Actually, all segments of simulated solutions $x(t)$ over an interval $t\in [s-\tau,s]$, i.e.\ states $x_{s}$, can be interpreted as further initial functions $\phi=x_{s}$. All these intermediate results have automatically been tested for zero convergence as well since
\begin{align}
\lim_{t\to\infty} \|x_t(\cdot; \phi(\theta;p))\|_X=& \lim_{t\to\infty} \|x_t(\cdot, x_{s}(\cdot; \phi(\theta;p)))\|_X.
\end{align}
As a consequence, it is worthwhile to further examine divergent solutions. Already Figure \ref{fig:scalar_step5} indicates that a diverging trajectory might temporarily get closer to the equilibrium than the initial function itself. 
Figure \ref{fig:Secondary} shows how an %\rev{n} 
oscillating solution segment with a smaller norm \rev{can} arise from a constant primary initial function. 
\begin{figure} 
\centering
\includegraphics[width=\linewidth]{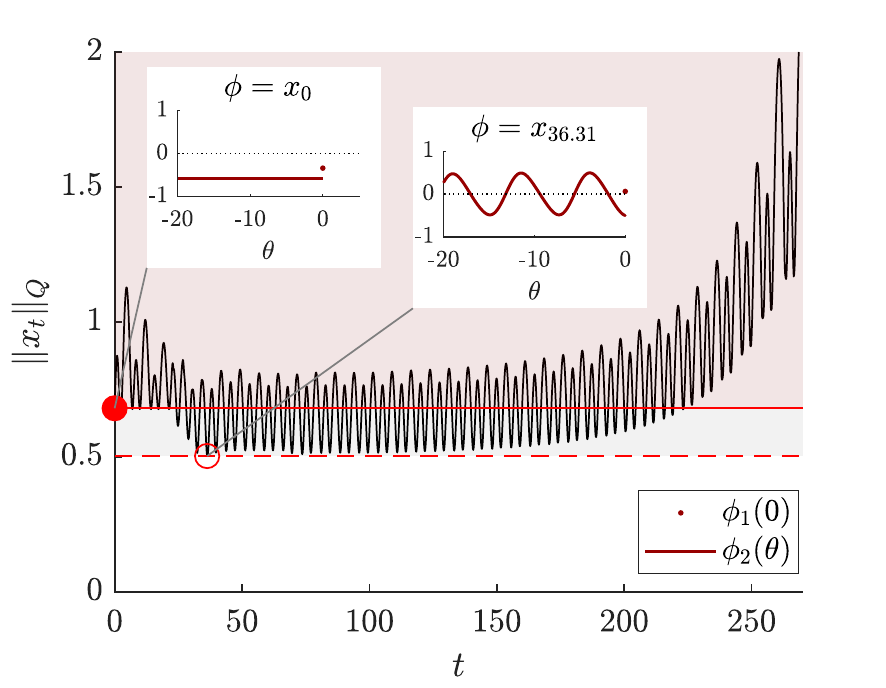}
\caption{\label{fig:Secondary}
Norm over time, i.e.\  $t\mapsto \|x_t\|_Q$, for a diverging solution of  (\ref{eq:delayedSwingEq}), (\ref{eq:param_swing_eq}).   
Not $\phi=x_0$ but a state $x_t$, $t>0$ takes the minimum norm value and thus is called secondary initial function.
The components $\phi=[\phi_1,\phi_2]^\top$ of both are shown in subfigures, where the point $\phi_1(0)$ fully represents the equivalence class of $\phi_1(\theta)$ in $X=Q$ as defined in (\ref{eq:QNorm}). 
\rev{The constant primary initial function $\phi=x_0$ corresponds to a red point %$p=[0.45,0.78]^\top$ 
in Figure \ref{fig:DOA_pixel_const} and Figure \ref{fig:p1_p2_plane}.}
}
\end{figure}
We denote the minimum norm value of diverging trajectories (red dashed line in Figure \ref{fig:Secondary}) as
\begin{align}
 %\hat R_X^{\mathscr T} 
\revi{\hat R_X^{\mathcal T} (\Phi)}\stackrel{\textrm{def}}=\min_{t\geq 0}\{\|x_t(\cdot;\phi)\|_X: x(t;\phi) \not \to 0, \phi \in \Phi \}, 
\label{eq:secondary}
\end{align}
where $\Phi$ is the set of primary initial functions (\ref{eq:allPhi}).
Obviously, without a noteworthy need of additional computational effort, 
%$\hat R_X^{\mathscr T}$ 
\revi{$\hat R_X^{\mathcal T} (\Phi)$} can only be equal or smaller than the estimation gained by primary initial functions $ \revi{\hat R_X(\Phi)}$. Indeed,  in Example \ref{ex:swingEq} \rev{below},  a tighter upper bound on the radius of attraction than the minimum norm of all considered primary initial functions will be achieved. 
\\
\\  
\subsection{Extended Primary Initial Functions} 
A {\itshape basis extension} can give further improvements of the estimation for the radius of attraction. Thereby, we mean to construct 
\begin{align}
\phi(\theta;p)=\sum_{i=1}^d p_i b_i(\theta)
\end{align}
with basis functions $b_i\in X([-\tau,0],\mathbb R)$ up to order $d>0$. The difference to primary initial functions (Section \ref{sec:primary}) lies in the higher number of parameters $[p_1,\ldots,p_d]=:p\in \mathbb R^{n\times d}$ and thus in the computational effort. 
This ansatz is also taken by Leng et al.  \cite{Leng.2016} who are not interested in an estimation for $R_X$ but in basin stability. The latter is up to normalization approximated by the percentage of those randomly chosen parameter values $p$ in a predefined set $p\in \mathcal B_d \subset \mathbb R^{n \times d}$ \rev{that} lead to a zero-convergent solution (Section \ref{sec:Introduction} (ii)). 
For instance, trigonometric, Legendre or Bernstein polynomials can be used as basis functions.
Of course, the results allow again to derive secondary initial functions.
\rev{
\begin{remark}[Approximation properties]
Similarily to most inner approximations of domains of attraction, no statement about the tightness of the derived upper bound can be given. Neither a statement about convergence for extended primary initial functions is possible.
Indeed, arbitrarily close approximations of initial functions in $X=C([-\tau,0],\mathbb R^n)$ are achievable if  polynomials with sufficiently high degree are considered (Stone-Weierstrass theorem). Moreover, if continuous dependence on the initial conditions is ensured (Hale and Verduyn Lunel \cite{Hale.1993}, Ch.2, Theorem 2.2), then a sufficiently close approximation of an element in $\mathcal D_X$ will also reach the ball $\mathcal B_X(0;\delta_a)$ of (\ref{eq:attractive}) in finite time.  Consequently, zero-convergence follows and thus the approximation also belongs to $\mathcal D_X$. However, this argument only holds for elements in $\mathcal D_X$, while the radius of attraction is determined by those initial functions \rev{that} lie on the boundary $\partial \mathcal D_X$. Asymptotic behavior of an element in this boundary set might not be reflected by any approximation. Hence, in view of the arbitrary difference between the radius of attraction $R_X$ and the potentially larger radius of attraction $R_X^{qs}$ w.r.t.\ the quasi-stability domain $\textrm{int}(\overline {\mathcal D_X})$ (Definition \ref{def:gen_domains}), even a large set of tested initial functions might not result in tight estimations of the radius of attraction. 
In the end, the combination of an analytically derived lower bound $\check R_X$ and a numerically derived upper bound $\hat R_X$, such that \ $R_X\in[\check R_X, \hat R_X]$, is most meaningful. The gap $\hat R_X-\check R_X$ provides a tightness bound for both estimations.
\end{remark}
}

\section{Estimations Based on Bifurcation Analysis}\label{sec:BifApproach}
A positively invariant or even invariant set with some distance to the attractor cannot be contained in the domain of attraction and thus yields an upper bound on the radius of attraction. 
\begin{proposition}
[\label{prop:invSet_notInDOA}Upper Bound on the Radius of Attraction Based on a Positively Invariant Set]
Assume $S$ is a positively invariant set and $\overline S \cap \{0_{[-\tau,0]}\} = \emptyset$. Then $S \cap \mathcal D_X=\emptyset$ and 
\begin{align}
R_X\leq\hat R_X^S:=\inf_{\phi\in S} \|\phi\|_X.
\end{align}
\end{proposition}
\begin{proof} 
If $\phi \in S$ is an element of a positively invariant set $S$ with some distance to the attractor in $\phi_e=0_{[-\tau,0]}$, the defining property of the domain of attraction $\lim_{t\to\infty} \|x_t(\theta;\phi)\|_X=0 $ cannot be fulfilled.
\end{proof}
\noindent In the approaches described above, non-zero-convergent forward trajectories from time\rev{-}domain simulations represent this set $S$. Furthermore, these forward trajectories approach $\omega$-limit sets, which are themselves invariant sets (complete continuity of $F$ in (\ref{eq:RFDE}) and boundedeness of trajectories presumed; Smith \cite{Smith.2011}, Corollary 5.6). However, there might also be equilibria, limit cycles, invariant tori, or chaotic sets, which are hidden \cite{Dudkowski.2016} 
in a certain sense or unstable and thus will not be approached in time\rev{-}domain simulations. However, it is challenging to locate these limit sets and even to prove their existence. 
In the following, we try to get a priori knowledge of some of these sets. 
\begin{remark}[Number of Limit Cycles]
It should be noted that the procedure described below will not necessarily capture all limit cycles. Even in the supposedly simple case of planar ODEs with polynomial vector fields of degree $d>1$ no finite upper bound for the number of limit cycles has yet been found. This is known as Hilbert's 16th problem \cite{Ilyashenko.2002}. 
\end{remark}
\noindent
We propose to start with the delay-free, \rev{thus} finite dimensional and manageable system \rev{that} results in setting %the delay parameter 
the delay \revi{value or the coefficient of the delayed terms} to zero. 
All additional equilibria, limit cycles, invariant tori, or chaotic sets have to emerge with an increase of 
\revi{these parameters}. %delay parameter. 
Hence, we consider 
the delay parameter or the respective coefficient 
as bifurcation parameter in a bifurcation analysis (for bifurcation theory of RFDEs see \cite{Hale.1993,Diekmann.1995,Kazarinoff.1978,Hassard.1981,Janssens.2010}). 
It is well known, that the loss and regain of stability by parameter variation is generically accompanied by a Hopf bifurcation. A subcritical Hopf bifurcation means that an unstable limit cycle occurs for \revi{parameter} values at which the equilibrium is stable (Figure \ref{fig:limit_cycles}). In a finite dimensional system in the plane, the situation would be clear: the unstable limit cycle, which surrounds the asymptotically stable equilibrium, represents the boundary of the domain of attraction. The infinite dimensionality of time\rev{-}delay systems makes conclusions harder. However, by Proposition \ref{prop:invSet_notInDOA}  the knowledge about such limit sets can be used for estimations of the radius of attraction.
\begin{figure}
\centering
\subfloat[\label{fig:limit_cycle_x0}Initial state $\phi\in Q$ and limit cycle solution 
$x(t)={[x_1(t),x_2(t)]^\top}$ 
for $\tau=20$ with period $T=7.4$.]{
\includegraphics[width=0.97\linewidth]{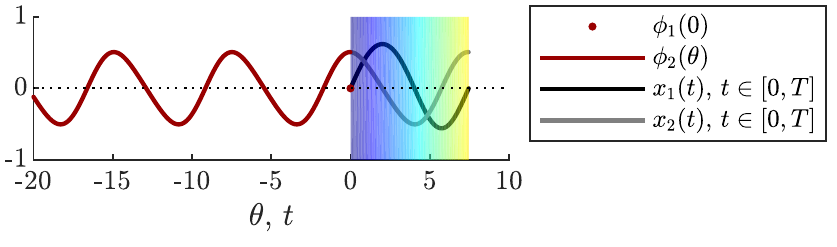}
}
\\
\subfloat[At delay values \rev{that} mark a regain of stability, subcritical Hopf bifurcations occur. For a certain delay value, the resulting limit cycle becomes visible in the $(x_1,x_2)$ plane. Coloring of limit cycles indicates time $t$ modulo period $T$.  Emerging limit cycles are only shown as long as the equilibrium is attractive.]{
\includegraphics[width=0.97\linewidth]{fig_limit_cycles.pdf}
}
\caption{\label{fig:limit_cycles} Unstable limit cycle in (\ref{eq:delayedSwingEq}), (\ref{eq:param_swing_eq}).}%Example \ref{ex:swingEq}. }
\end{figure}
\\
\\
As in the ODE case, it is possible to formulate a boundary value problem (BVP) that is solved by the periodic solutions of a limit cycle. 
Periodicity of the whole solution $x(t)=x(T+t), t \in [-\tau,\infty)$ is by semigroup property already ensured, if $x(\theta)=x(T+\theta), \theta \in[-\tau,0]$, i.e.\ by equality of the two states $x_0=x_T$. Hence, instead of an initial condition $x_0=\phi$ the boundary condition $x_0=x_T$ is required, such that a two-point BVP over one period $t\in(0,T]$ has to be solved. The dependence of the boundary position on the beforehand unknown parameter $T$ can be removed by a time scale transformation. Consider $\dot x=f(x(t),x(t-\tau))$. With $t={\hat t} T$ and $x(\hat t T)=:z(\hat t)$ the two-point BVP over the normalized  period $\hat t\in(0,1]$ becomes  
\begin{subequations}\label{eq:BVP}
\begin{align}
\frac{\textrm d z}{\textrm d \hat t}(\hat t)&=Tf\left(z(\hat t),z\left(\hat t-\frac \tau T\right)\right), &&\hat t \in(0,1] \\
z({\hat \theta})&=z(1+{\hat \theta}), &&{\hat \theta}\in\left(-\frac \tau T,0\right].
\end{align}
\end{subequations}
Still, it is ambiguous which delay-spanning segment of the periodic solution is addressed. This translational symmetry has to be resolved by an additional phase shift condition leading to a unique solution.
\rev{However, the above boundary value problem can only be solved numerically if an appropriate initial estimation of the limit cycle is given. To this end, we consider a numerical bifurcation analysis with time delay as bifurcation parameter. Usually, the latter is used if parameter variations in a given range are of interest. However, as a byproduct, effects of parameter variation on domains of attraction might become visible \cite{Stepan.1991,Molnar.2017,Molnar.2018}. We are mainly interested in a fixed parametrization, which at a first glance is not related to bifurcation analysis. However,}
since in a bifurcation analysis limit cycles are tracked from their emergence, (\ref{eq:BVP}) is well initializable \rev{in an iterative manner based on successive \revi{parameter} steps until the parametrization of interest is reached.}

We make use of an implementation in the Matlab toolbox \texttt{DDE-BIFTOOL} \cite{Sieber.2014}, which is based on a collocation method with piecewise polynomials \cite{Roose.2007}. The latter provides a numerical approximation of the periodic solution. 
\\
\\
\rev{Already for simple lower order systems, the infinite dimensionality of time\rev{-}delay systems frequently results in complex phenomena and a rich variety of bifurcations \cite{Sprott.2007,Lakshmanan.2010,Molnar.2017,Molnar.2018,Scholl.2019b}.} Of course, \rev{bifurcations of the above described} limit cycles might occur with further delay parameter variation. For instance, limit cycles vanish if they collide in a fold bifurcation, a new limit cycle might branch off in a period doubling bifurcation, or an invariant torus, which surrounds a limit cycle, might stem from a Neimark-Sacker bifurcation \cite{Kuznetsov.1998}. A bifurcation analysis aims at detecting such phenomena. Approximations of Floquet multipliers allow to conclude stability. After the numerical analysis, if a certain fixed delay value is given, some occurring invariant sets can be concluded. The periodic solution leads to the limit cycle trajectory (set of all delay-width solution segments) \rev{that} we use as invariant set in Proposition \ref{prop:invSet_notInDOA}. Even more, the bifurcation analysis results in valuable insights how the system behavior changes with increasing delay. Thereby, also the variation of the derived upper bound for the radius of attraction with increasing delay values can be observed.
\begin{remark} [Stable Manifolds]\label{rem:StableManifold}
Generically, unstable equilibria or unstable limit cycles in RFDEs are saddle-type sets with a finite number of unstable eigenvalues or Floquet multipliers, respectively, but an infinite number of stable ones. Hence, there is also a stable manifold corresponding to the stable eigenfunctions of the linearized system \cite{Diekmann.1995}. Indeed, the stable manifold of a saddle-type set or the domain of attraction of a competing attractor is also a possible invariant set in Proposition \ref{prop:invSet_notInDOA}. However, these sets remain usually unknown. Nevertheless, knowledge from ODEs suggests that solutions nearby the stable manifold approach the saddle-type set temporarily in time\rev{-}forward simulations. As a result, the saddle-type set can to some exten\rev{t} already be reflected in secondary initial functions. This conjecture is supported by comparing the secondary initial function in Figure \ref{fig:Secondary} and the limit cycle solution in Figure \ref{fig:limit_cycle_x0}.
\end{remark}
\rev{
\begin{remark}[Relation to Bi- or Multistability]
A coexisting second attractor is an example of an invariant set in Proposition \ref{prop:invSet_notInDOA} and thus implies  that the first attractor has a finite radius of attraction $R_X<\infty$ and cannot be globally asymptotically stable. 
The coexistence of two or more attractors is commonly denoted as bistability or multistability, respectively \cite{Broer.2011,Losson.1993,Kuznetsov.2014b,Efimov.2015,Balanov.2005}. Thereby, it is a matter of definition whether - based on a space compactification - infinity can also be considered as an attractor in a non-dissipative system. All unbounded trajectories form its domain of attraction (without a meaningful radius of attraction).
%. It should be noted that a corresponding radius of attraction is meaningless, since it would be inevitably infinite no matter whether the domain of attraction is indeed global.
Conversely, non-global asymptotic stability of an equilibrium  (i.e.\ $R_X<\infty$) does not imply the existence of a second attractor (including an attractor at infinity). To this end, consider the scalar ODE
$\{\dot x = -x $, if $x\leq 1$ and $\dot x = 0 $, if $x>1\}$, which has only one attractor at $x_e=0$ with domain of attraction $(-\infty,1)$ and radius of attraction $R_{\mathbb R}=1$ (the other equilibria are non-isolated and thus non-attractive).
\end{remark}
}
%%%%%%%%%%%%%%%%%%%%%%%%%%%%%%%%%%%%%%%%%%%%%%%%%%%%%%%%%%%%%%%%%%%%%%
\section{Demonstrative Example}
The following example demonstrates the proposed approach for estimating the radius of attraction. 
Second order ordinary differential equations with periodic nonlinearity are relevant for various non-related applications like mechanical pendulum mechanisms, Josephson junction circuits, phase-locked loops, or synchronous generators  \cite{Khalil.2002}. At least in the latter case, the impact of delayed damping on the domain of attraction is of high technical interest \cite{Schafer.2015}.  Schaefer et al. \cite{Schafer.2015} give a simulative analysis of basin stability for such a delayed swing equation.
\begin{example}[Swing Equation with Time Delay]\label{ex:swingEq}
Consider the swing equation or constantly driven pendulum equation with additional delayed damping \rev{described by $\ddot y(t)+a  \dot y(t)+\tilde a \dot y(t-\tau) + \sin(y(t))=w$. A state space representation in $x=[y-y_e,\dot y]$, such that the equilibrium  $y_e:=\arcsin(w)$ is shifted to the origin $x_e=[0,0]^\top$, reads}
\begin{alignat}{5}
\dot x(t)=& %\underbrace
\begin{bmatrix} 0 & 1  \\ 0 & -a\end{bmatrix}
x(t) 
+
\begin{bmatrix} 0 & 0  \\ 0 & -\tilde a \end{bmatrix} 
x(t-\tau)
+&\begin{bmatrix} 0 \\ w-\sin(x_1(t)+y_e)\end{bmatrix}&
\nonumber\\
x(\theta) &=\phi(\theta), & \quad \theta \in[-\tau,0]&
\label{eq:delayedSwingEq}
\end{alignat}
\rev{with} $x(t)\in \mathbb R^2$, $a,\tilde a,\tau \in \mathbb R_{>0}$, $w\in(0,1)$,  and $\phi\in X$ in a Banach space $X$ to be selected. Unless otherwise stated, results refer to the parameter set
\begin{align}
a=0.05,\quad \tilde  a=0.125,\quad w=0.5,\quad \tau=20.
\label{eq:param_swing_eq}
\end{align}
The aim is to find an upper bound $\hat R_X> R_X$ for the radius of attraction $R_X$ in the chosen norm.
\\
\\
{\itshape (i) Selection of the State Space $X$} \\
The history of the first component $x_1$ does not contribute to the system dynamics at all. According to (\ref{eq:QuotientSpace}), this motivates the use of a quotient space $X=Q$. We define
\begin{align}
Q&:=PC([-\tau,0],\mathbb R^2) / \{\phi \in PC: \phi_1(0)=0, \phi_2\equiv 0\}, \nonumber \\
\|\phi\|_{Q}&:=\sqrt{\vert \phi_1(0) \vert ^2+\|\phi_2\|_C^2}.
\label{eq:QNorm_PC}
\end{align}
\begin{remark} [Definition on a Manifold]
There is a rotational symmetry in $x_1$, such that -
instead of $x(t)\in \mathbb R^2$ - the instantaneous values can also be defined on a cylinder, i.e.\ $[x_1(t),x_2(t)]^\top \in S^1 \times \mathbb R$ and consequently
\begin{align}
Q_{cyl}:=PC([-\tau,0],S^1\times \mathbb R) / \{\phi \in PC: \phi_1(0)=0, \phi_2\equiv 0\}. \nonumber
\end{align}
Thereby, the number of stable equilibria reduces from infinite to one in Example \ref{ex:swingEq} since all belong to the same equivalence class. 
Stability theory of RFDEs can be transferred to such systems on manifolds \cite{Hale.1993,Oliva.1969}. (The respective far reaching consequences are not in the scope of this paper.)
\end{remark}
~\\
{\itshape (ii) Selection of Primary Initial Functions}
\\
Actually, a physical meaning of $x_1$ and $x_2$ as angle and angular velocity in (\ref{eq:delayedSwingEq}) implies $\frac{\textrm d}{\textrm d t}\phi_1=\phi_2$ (Remark \ref{rem:scalarHigherOrder}). However, because of the irrelevance of the $x_1$-history in the above introduced quotient space (\ref{eq:QNorm_PC}), this relation has no consequence for the selection of initial functions.
As representative of the first component $\phi_1\colon[-\tau,0]\to\mathbb R$ we use
\begin{align}
\label{eq:phi1p1}
\phi_1(\theta)\equiv \phi_1(0)=p_1, 
\end{align}
where $p_1\in \mathbb R$ parametrizes this initial function. 
For the second component $\phi_2\colon[-\tau,0]\to\mathbb R$, we select a finite number of primary initial functions  parameterized by $p_2\in \mathbb R$ (Section \ref{sec:primary}):
\begin{subequations}
\begin{alignat}{4}
\phi_2^1(\theta;p_2)&\equiv p_2 \\
\phi_2^2(\theta;p_2)&=\left\{\begin{array}{cl} 0, & -\tau\leq\theta<0 \\ p_2, & \theta=0 \end{array}\right. \label{eq:Initial_stepFunctions} \\
\phi_2^3(\theta;p_2)&=p_2 \cos\left(\omega \theta\right),&& \omega \in \mathbb R\\ %\frac{2\pi}{\tau}
\phi_2^4(\theta;p_2)&=p_2 \sin\left(\omega \theta\right),&& \omega \in \mathbb R.
\end{alignat}
\end{subequations}
In order to simplify the evaluation, the primary initial functions above are chosen such that $\|\phi_2^k(\,\cdot\,;p_2)\|_C=\vert p_2 \vert$, $k\in\{1,\ldots,4\}$.
This parametrization is advantageous since $\|\cdot\|_{Q}$ norms of the initial functions $\phi\in \Phi^k=\{[\phi_1,\phi_2^k]^\top\}_k$  equal  Euclidean norms in the $(p_1,p_2)$ plane, i.e.\
\begin{align}
\|\phi(\,\cdot\,;[p_1,p_2]^\top)\|_{Q}=\|[p_1,p_2]^\top\|_2.
\label{eq:QNorm_EuclideanNorm}
\end{align}
Thus, a norm ball $\mathcal  B_Q(r){=}\{\phi\in Q: \|\phi\|_{Q}< r\}$ 
corresponds to the inner of a circle in the parameter plane. 
\begin{remark}[Case $\tau\to 0$]
If the delay $\tau \to 0$ vanishes, the domain of $\phi$  reduces to a singleton $[-\tau,0]\to \{0\}$ and the state space $Q$ is replaced by $\mathbb R^2$, which is fully covered by the two values $[p_1,p_2]^\top=[\phi_1(0),\phi_2(0)]^\top=[x_1(0),x_2(0)]^\top$ describing $\Phi$.  Then  $\|\cdot\|_{Q}$ is replaced by $\|\cdot \|_2$, such that $\hat R_2=R_2$ describes exactly the maximum inner circle $\mathcal B_2(R_2)$ within the domain of attraction of the ODE in the phase plane.
\end{remark}
~\\
{\itshape (iii) Simulative Results for Primary Initial Functions} \\
The domain of attraction for the $k$-th family of initial functions (\ref{eq:initFuncFamily}) is given by
\begin{align}
%\mathcal D_{\Phi}^k  \stackrel{\textrm{def}}
(\mathcal D_Q \cap \Phi^k)=\{\phi\in\Phi^k: \lim_{t\to \infty} x(t;\phi) = 0 \} .
\end{align}
 Figure \ref{fig:ROA_pixel} shows results in the $(p_1,p_2)$ parameter plane. To get an approximation of the relevant parts of the 
boundary $\partial (\mathcal D_Q \cap \Phi^k)$, 
parameters are increased in a star-like scheme with a finite number of equidistant angles, until convergence to the origin fails. 
For visualization purposes, Figure \ref{fig:p1_p2_plane} contains the (possibly non-convex) hull of the outer parameter values for each family of initial functions. 
\begin{figure}
\centering
\includegraphics[width=\linewidth]{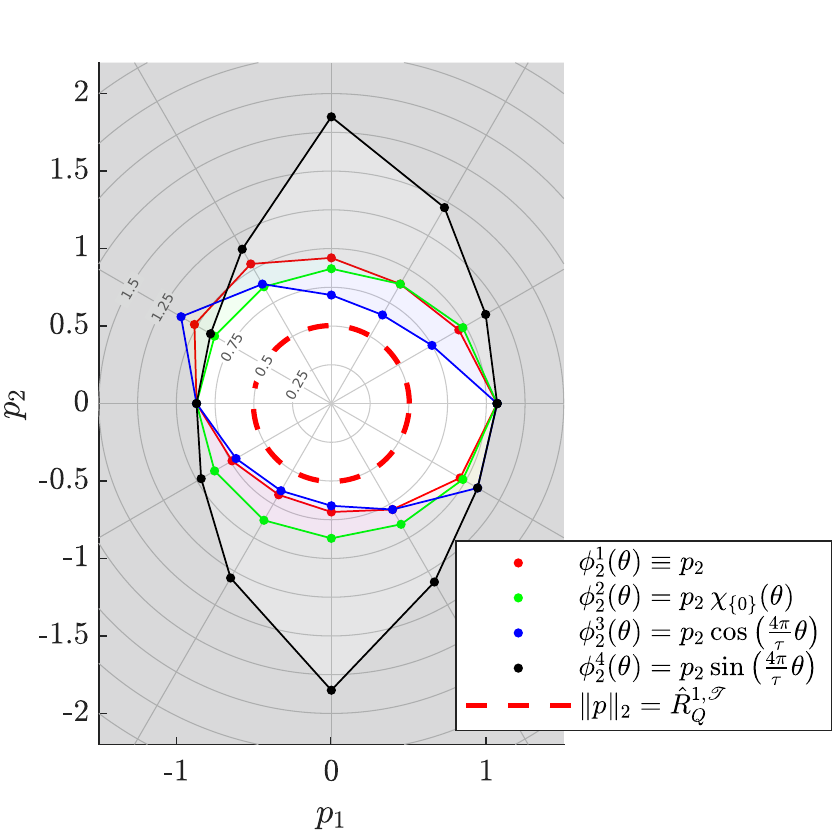}
\caption{\label{fig:p1_p2_plane}  Points mark smallest $(p_1,p_2)$ combinations along the investigated directions \rev{that} lead to non-zero-convergent solutions in (\ref{eq:delayedSwingEq}), (\ref{eq:param_swing_eq}). %Example \ref{ex:swingEq}. 
The primary initial functions $\phi(\theta; [p_1,p_2]^\top)=[p_1,\phi_2^k(\theta;p_2)]^\top$ correspond to \rev{the four} subfigures in Figure \ref{fig:ROA_pixel}.}
\end{figure}
Obviously, not the inner circle of the visualized hull but the minimum norm of the single vertices results in the upper bound
\begin{align}
\hat R^{k}_Q :=
\revi{\hat R_Q(\Phi^k)}=\min\{\|\phi\|_Q: x(t;\phi) \not \to 0, \phi \in \Phi^k\}.
\end{align}
%Thereby, the superscript $k$ indicates that  $\hat R^{k}_Q$ refers to the $k$-th family of initial functions $\Phi^k$, $k\in\{1,\ldots,4\}$. 
All four families of initial functions $\Phi^k$, $k\in\{1,\ldots,4\}$ taken in  Figure \ref{fig:p1_p2_plane} lead to 
$\min_k \hat R^k_Q=\hat R^3_Q=0.65$. The latter corresponds to the most central blue point.
\\
\\
{\itshape (iv) Secondary Initial Functions} \\
By consideration of secondary initial functions, as described in Section \ref{Sec:SecondaryInitialFunctions}, an improvement of this estimation is possible without the need for further simulations.
Instead of only considering the initial states $x_0=\phi\in \Phi^k$, 
the estimation $\hat R_Q^{k,\mathscr T}:=\hat R_Q^{\mathscr T}(\Phi^k)$ uses the minimum norm of all solution states $x_t(\cdot \,; \phi)$, $t\geq 0$ \rev{that} result from non-zero-convergent primary initial functions $\phi\in \Phi^k$. 
The final estimation 
$
\hat R_Q^{\mathscr T} = \min_k \hat R_Q^{k,\mathscr T}
$
merges the results of all families of initial functions. The dashed line in Figure \ref{fig:p1_p2_plane} indicates that there are indeed diverging solutions \rev{that} come even closer to the attractor than all the considered primary initial functions.  
Secondary initial functions of all chosen families $\Phi^k$, $k\in\{1,\ldots,4\}$ yield similar minimum norm values $\hat R_Q^{k,\mathscr T}$. The best result is gained by a constant initial function, which leads to a diverging solution with minimum Q-norm $\hat R^{\mathscr T}_Q = \hat R^{1,\mathscr T}_Q =0.503$, as shown in Figure \ref{fig:Secondary}.   
\\
\\
{\itshape (v) Bifurcation Analysis}
\begin{figure}
\centering
\includegraphics[width=\linewidth]{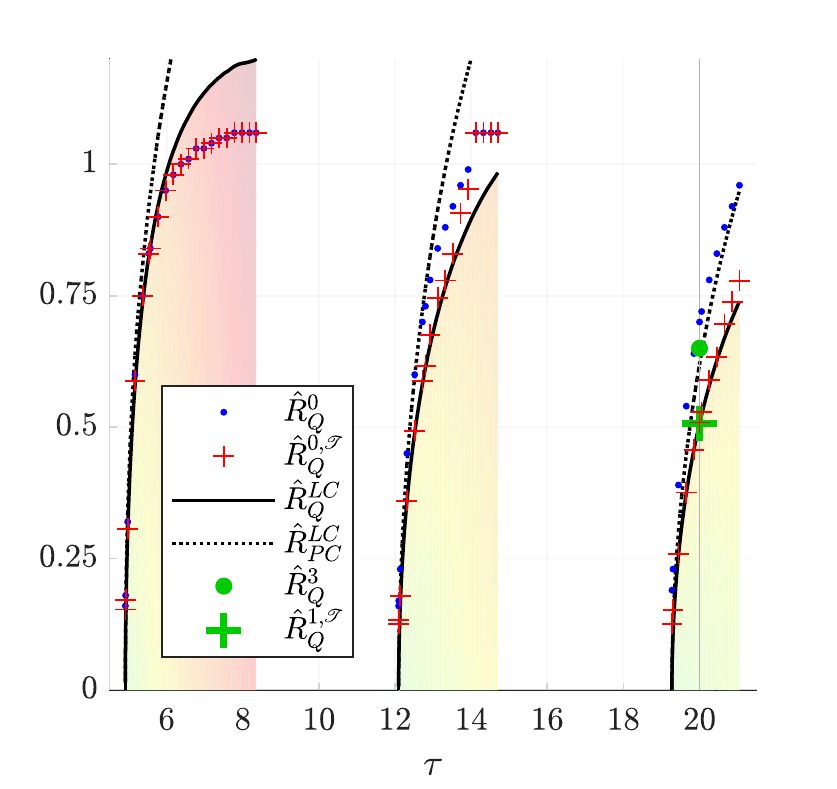}
\caption{\label{fig:limit_cycles_norms} Overestimations for the radius of attraction in  (\ref{eq:delayedSwingEq}), (\ref{eq:param_swing_eq}) with different delay values $\tau$. The background refers to a side view of Figure \ref{fig:limit_cycles}. Results are gained by limit cycle states $x_t^{LC}$, the simple  primary initial functions $\Phi^0=\{\phi(\theta)\equiv[0,-p]^\top, p\in \mathbb R_{>0}\}$, and corresponding secondary initial functions. Results from $(iii)$ and $(iv)$ for $\tau=20$ are marked in green.  
%Hence, $\hat R_X^{LC}=\min_{t\in[0,T)} \|x_t^{LC}\|_Q$, $\hat R_Q^0=\|\phi^0_{\textrm{crit}}\|_Q$ and  $\hat R_Q^{0,\mathcal T}=\min_{t\in[0,\infty)} \|x_t(\,\cdot\,; \phi^0_{\textrm{crit},\mathscr T})\|_Q$. 
}
\end{figure}
\\ 
The trivial equilibrium is alternately stable and unstable with increasing delay. A bifurcation analysis of (\ref{eq:delayedSwingEq})  reveals that each regain of stability is accompanied by a subcritical Hopf bifurcation \rev{(cmp.\ \cite{Scholl.2019})}. Vertical slices in Figure \ref{fig:limit_cycles} show the emerging unstable limit cycle in the $(x_1,x_2)$ plane for different delay values. For the chosen delay value $\tau=20$ the corresponding periodic solution in Figure \ref{fig:limit_cycle_x0} is remarkably similar to the state \rev{that} was gained as secondary initial function in Figure \ref{fig:Secondary}. This can be explained by the saddle character of the limit cycle (Remark \ref{rem:StableManifold}). A comparison of the norm values is given in Figure \ref{fig:limit_cycles_norms}. Thereby, the role of the the quotient space norm becomes visible. The radius of attraction $R_{PC}$ in the space $PC([-\tau,0],\mathbb R^2)$ is identical to $R_{Q}$ in the quotient space. However, the uniform norm of limit cycle states would lead to unnecessarily conservative estimations $\hat R^{LC}_{PC}$ (Figure \ref{fig:limit_cycles_norms}). The minimum $Q$-norm value of the limit cycle states $x_t^{LC}$ is $\hat R^{LC}_Q:=\min_{t\in[0,T)} \|x_t^{LC}\|_Q = 0.507$. 
\\
\\
Hence, for the delay value under consideration, secondary initial functions and the bifurcation analysis based approach yield almost equal upper bounds $\hat R^{LC}_Q\approx \hat R^{\mathscr T}_Q$ on the radius of attraction $R_Q\leq \hat R_Q=:\hat R^{\mathscr T}_Q$ in Example \ref{ex:swingEq}. Obviously, this does not always hold true. For small delay values in Figure \ref{fig:limit_cycles_norms},  already the chosen simple primary initial functions lead to better estimations. This means \rev{that} primary initial functions with smaller norm values than the limit cycle lead to diverging solutions. Hence, these must cross the limit cycle in the $(x_1,x_2)$ plane. Time delay enables such crossings since forward uniqueness in RFDEs does not refer to instantaneous values $x(t)\in \mathbb R^2$ but to solution segments, i.e.\ to states $x_t\in Q$ in the infinite dimensional space $X=Q$. Additionally, Figure \ref{fig:limit_cycles_norms} illustrates that secondary initial functions can only lead to improvements. %And indeed they do. 
However, Figure \ref{fig:limit_cycles_norms} also demonstrates that the bifurcation analysis based approach is a worthwhile supplement. Compared to the search for diverging primary initial functions in time\rev{-}domain simulations, the numerical solution of boundary value problems for limit cycles is less computationally expensive. In addition, the bifurcation analysis gives further insights.
The growing limit cycles, which are revealed by a bifurcation analysis, result in an enlargement of the upper bound on the radius of attraction with increasing delay. 
\end{example}

\section{Conclusion}
Though as important as in delay-free systems, domains of attraction in time\rev{-}delay systems still find relatively little consideration in literature. The previous sections contribute the following:
\begin{itemize}
\item {\itshape Estimating Domains of Attraction: \\Approaches and Quantification} \\
Various requirements concerning the domain of attraction are important in practical applications. An overview of possible approaches for estimations of the domain of attraction in time\rev{-}delay systems is given. 
The radius of attraction in a certain Banach space represents an easy to handle number. 
In many cases a lower bound on the radius of attraction is of interest. These estimations are mostly conservative and will be subject of further considerations. Nevertheless, such results leave unavoidably open whether a small radius means a conservative estimation or indeed a small domain of attraction. 
In order to prove the latter, upper bounds on the radius of attraction become important. 
These can be obtained by simulative means. 
\item {\itshape Generalized Concepts for the Domain of Attraction: \\ No Lebesgue Measure in Infinite Dimensional Spaces} \\
It should be kept in mind that simulation usually does only cover generic cases. In the context of ordinary differential equations, genericity is addressed by concepts like quasi-stability \rev{domains} or almost global stability. The paper proposes to generalize these definitions to the infinite dimensional case, thereby the fact that there is no Lebesgue measure in infinite dimensional spaces must be taken into account.
\item {\itshape Selection of the Banach Space: \\Consequences for the Radius of Attraction} \\
The radius of attraction and its physical interpretation depend heavily on the chosen norm. 
It is pointed out that the radius of attraction in some spaces is not invariant under time scale transformations.
\item {\itshape Quotient Spaces: \\ Incorporate that States are Differently Delayed} \\
We propose to choose a quotient space if some of the state variables do not occur delayed in the system equations. Otherwise unused initial data leads to unnecessary restrictions. 
\item {\itshape Secondary Initial Functions: \\ Improvement of the Examined Set of Initial Functions } \\ An important question is how to choose the initial functions \rev{that} are \rev{numerically} tested for belonging to the domain of attraction. Besides physically or technically motivated functions, primary functions in the sense of simple function families and the construction of initial functions with more degrees of freedom by certain basis functions are convenient. The paper proposes to take into account all available solution segments from simulations as well. These secondary initial functions are able to improve estimations for the radius of attraction without notable computational effort.
\item {\itshape Bifurcation Analysis: \\ Getting A Priori Knowledge of Other Invariant Sets} \\ 
Competing equilibria, limit cylces, invariant tori, chaotic sets\rev{,} or more general invariant sets with some distance to the attractor can also contribute to upper bounds on the radius of attraction. While the delay-free system might be well understood, time delay is able to induce complex dynamics. In order to detect some further limit sets, we propose a numerical bifurcation analysis. By starting from a delay-free system, a bifurcation analysis allows first\rev{ly} to detect their emergence and then to track them numerically.  
\item {\itshape Demonstrative Example: \\ Swing Equation with Additional Delayed Damping} \\
An example demonstrates significant improvements in the estimation of the radius of attraction by secondary initial functions. 
It also confirms that the proposed bifurcation analysis based approach is a worthwhile additional element \rev{for upper bound estimations of the radius of attraction}. 
\end{itemize}

\section*{Acknowledgement}
The authors acknowledge support by the Helmholtz Association under the Joint Initiative "Energy Systems Integration" (ZT-0002).

\bibliographystyle{ieeetr}
\bibliography{Literatur}

\begin{thebibliography}{10}

\bibitem{Gu.2003}
K.~Gu, V.~L. Kharitonov, and J.~Chen, {\em Stability of time-delay systems}.
\newblock Boston: Birkh{\"a}user, 2003.

\bibitem{Niculescu.2004}
S.-I. Niculescu and K.~Gu, eds., {\em Advances in time-delay systems}.
\newblock Berlin and Heidelberg: Springer, 2004.

\bibitem{Briat.2015}
C.~Briat, {\em Linear parameter-varying and time-delay systems: Analysis,
  observation, filtering {\&} control}.
\newblock Heidelberg: Springer, 2015.

\bibitem{Wu.2010}
M.~Wu, Y.~He, and J.-H. She, {\em Stability analysis and robust control of
  time-delay systems}.
\newblock Berlin, Heidelberg: Springer, 2010.

\bibitem{Michiels.2014}
W.~Michiels and S.-I. Niculescu, {\em Stability, control, and computation for
  time-delay systems: An eigenvalue-based approach}.
\newblock Philadelphia: {SIAM Soc. for Indust. and Appl. Math}, 2014.

\bibitem{Hale.1993}
J.~K. Hale and S.~M. {Verduyn Lunel}, {\em Introduction to functional
  differential equations}.
\newblock New York: Springer, 1993.

\bibitem{Seuret.2016}
A.~Seuret, F.~Gouaisbaut, and L.~Baudouin, ``D1.1 - {O}verview of {L}yapunov
  methods for time-delay systems: Rapport laas no. 16308,'' {\em HAL
  archives-ouvertes.fr}, no.~hal-01369516, 2016.

\bibitem{Fridman.2014b}
E.~Fridman, ``Tutorial on {L}yapunov-based methods for time-delay systems,''
  {\em Eur. J. Control}, vol.~20, no.~6, pp.~271--283, 2014.

\bibitem{Niculescu.2001}
S.-I. Niculescu, {\em Delay effects on stability: A robust control approach}.
\newblock Berlin: Springer, 2001.

\bibitem{Halanay.1966}
A.~Halanay, {\em Mathematics in science and engineering: Differential
  equations: Stability, oscillations, time lags}.
\newblock Elsevier, 1966.

\bibitem{Dambrine.1994}
M.~Dambrine, {\em Contribution {\`a} l'{\'e}tude de la stabilit{\'e} des
  syst{\`e}mes {\`a} retards}.
\newblock PhD thesis, {Universit{\'e} Lille1 - Sciences et Technologies}, 1994.

\bibitem{Briat.2011}
C.~Briat, {\em Robust stability analysis in the $\ast$-norm and
  {L}yapunov-{R}azumikhin functions for the stability analysis of time-delay
  systems: (CDC-ECC 2011) ; Orlando, Florida, USA, 12 - 15 December 2011}.
\newblock Piscataway, NJ: IEEE, 2011.

\bibitem{Breda.2015}
D.~Breda, S.~Maset, and R.~Vermiglio, {\em Stability of linear delay
  differential equations: A numerical approach with MATLAB}.
\newblock New York: Springer, 2015.

\bibitem{Engelborghs.2002}
K.~Engelborghs, T.~Luzyanina, and D.~Roose, ``Numerical bifurcation analysis of
  delay differential equations using {DDE-BIFTOOL},'' {\em ACM Trans. Math.
  Softw.}, vol.~28, no.~1, pp.~1--21, 2002.

\bibitem{Insperger.2017}
T.~Insperger, T.~Ersal, and G.~Orosz, eds., {\em Time delay systems: Theory,
  numerics, applications, and experiments}.
\newblock Cham: Springer, 2017.

\bibitem{Jarlebring.2008}
E.~Jarlebring, {\em The spectrum of delay-differential equations: numerical
  methods, stability and perturbation}.
\newblock Dissertation, {Technische Universit{\"a}t Carolo-Wilhelmina zu
  Braunschweig}, 2008.

\bibitem{Otto.2019}
A.~Otto, W.~Just, and G.~Radons, ``Nonlinear dynamics of delay systems: An
  overview,'' {\em Philosophical transactions. Series A, Mathematical,
  physical, and engineering sciences}, vol.~377, no.~2153, p.~20180389, 2019.

\bibitem{Dombovari.2008}
Z.~Dombovari, R.~E. Wilson, and G.~Stepan, ``Estimates of the bistable region
  in metal cutting,'' {\em Proceedings. Mathematical, physical, and engineering
  sciences}, vol.~464, no.~2100, pp.~3255--3271, 2008.

\bibitem{Yan.2019}
Y.~Yan, J.~Xu, and M.~Wiercigroch, ``Estimation and improvement of cutting
  safety,'' {\em Nonlinear Dynamics}, vol.~53, no.~2, p.~619, 2019.

\bibitem{Dombovari.2019}
Z.~Dombovari, A.~Iglesias, T.~G. Molnar, G.~Habib, J.~Munoa, R.~Kuske, and
  G.~St{\'e}p{\'a}n, ``Experimental observations on unsafe zones in milling
  processes,'' {\em Philosophical transactions. Series A, Mathematical,
  physical, and engineering sciences}, vol.~377, no.~2153, p.~20180125, 2019.

\bibitem{Schafer.2015}
B.~Sch{\"a}fer, M.~Matthiae, M.~Timme, and D.~Witthaut, ``Decentral smart grid
  control,'' {\em New Journal of Physics}, vol.~17, no.~5, p.~059502, 2015.

\bibitem{Diekmann.1995}
O.~Diekmann, S.~M. {Verduyn Lunel}, S.~A. Gils, and H.-O. Walther, {\em Delay
  equations: Functional-, complex-, and nonlinear analysis}.
\newblock New York: Springer, 1995.

\bibitem{MelchorAguilar.2006}
D.~Melchor-Aguilar and S.-I. Niculescu, ``Estimates of the attraction region
  for a class of nonlinear time-delay systems,'' {\em IMA J. Math. Control
  Inf.}, vol.~24, no.~4, pp.~523--550, 2006.

\bibitem{Villafuerte.2007}
R.~Villafuerte and S.~Mondi{\'e}, ``On improving estimate of the region of
  attraction of a class of nonlinear time delay system,'' {\em IFAC Proceedings
  Volumes}, vol.~40, no.~23, pp.~227--232, 2007.

\bibitem{Fridman.2014}
E.~Fridman, {\em Introduction to time-delay systems: Analysis and control}.
\newblock Cham: Springer, 2014.

\bibitem{Cao.2002}
Y.-Y. Cao, Z.~Lin, and T.~Hu, ``Stability analysis of linear time-delay systems
  subject to input saturation,'' {\em IEEE Transactions on Circuits and Systems
  I: Fundamental Theory and Applications}, vol.~49, no.~2, pp.~233--240, 2002.

\bibitem{Coutinho.2008}
D.~F. Coutinho and C.~E. de~Souza, ``Delay-dependent robust stability and
  {$L^2$}-gain analysis of a class of nonlinear time-delay systems,'' {\em
  Automatica}, vol.~44, no.~8, pp.~2006--2018, 2008.

\bibitem{Goldsztejn.2019}
A.~Goldsztejn and G.~Chabert, ``Estimating the robust domain of attraction for
  non-smooth systems using an interval {L}yapunov equation,'' {\em Automatica},
  vol.~100, pp.~371--377, 2019.

\bibitem{Chiang.2011}
H.-D. Chiang, {\em Direct methods for stability analysis of electric power
  systems: Theoretical foundation, BCU methodologies, and applications}.
\newblock Hoboken, New Jersey: Wiley, 2011.

\bibitem{Chiang.2009}
H.-D. Chiang and Y.~Tada, ``Design and implementation of on--line dynamic
  security assessment,'' {\em IEEJ Transactions on Electrical and Electronic
  Engineering}, vol.~4, no.~3, pp.~313--321, 2009.

\bibitem{Menck.2013}
P.~J. Menck, J.~Heitzig, N.~Marwan, and J.~Kurths, ``How basin stability
  complements the linear-stability paradigm,'' {\em Nature Physics}, vol.~9,
  no.~2, pp.~89--92, 2013.

\bibitem{Leng.2016}
S.~Leng, W.~Lin, and J.~Kurths, ``Basin stability in delayed dynamics,'' {\em
  Scientific reports}, vol.~6, p.~21449, 2016.

\bibitem{Hale.1965b}
J.~K. Hale, ``Sufficient conditions for stability and instability of autonomous
  functional-differential equations,'' {\em J. Differ. Equ.}, vol.~1, no.~4,
  pp.~452--482, 1965.

\bibitem{Khalil.2002}
H.~K. Khalil, {\em Nonlinear systems}.
\newblock Upper Saddle River: {Prentice Hall}, 2002.

\bibitem{Souza.2014b}
C.~E. de~Souza and D.~Coutinho, ``Delay-dependent regional stabilization of
  nonlinear quadratic time-delay systems,'' {\em IFAC Proceedings Volumes},
  vol.~47, no.~3, pp.~10084--10089, 2014.

\bibitem{Haddock.1983}
J.~R. Haddock and J.~Terj{\'e}ki, ``{L}iapunov-{R}azumikhin functions and an
  invariance principle for functional differential equations,'' {\em J. Differ.
  Equ.}, vol.~48, no.~1, pp.~95--122, 1983.

\bibitem{Liu.2014}
K.~Liu and E.~Fridman, ``Delay-dependent methods and the first delay
  interval,'' {\em Systems {\&} Control Letters}, vol.~64, pp.~57--63, 2014.

\bibitem{Kharitonov.2013}
V.~L. Kharitonov, {\em Time-delay systems: Lyapunov functionals and matrices}.
\newblock New York: {Birkh{\"a}user Springer}, 2013.

\bibitem{Bellman.1963}
R.~Bellman and K.~L. Cooke, {\em Differential-difference equations}.
\newblock {Rand Corporation}, 1963.

\bibitem{Hahn.1967}
W.~Hahn, {\em Stability of Motion}.
\newblock New York: Springer, 1967.

\bibitem{Kloeden.2011}
P.~E. Kloeden and M.~Rasmussen, {\em Nonautonomous dynamical systems}.
\newblock Providence: {American Mathematical Society}, 2011.

\bibitem{Hinrichsen.1990b}
D.~Hinrichsen and A.~J. Pritchard, ``Real and complex stability radii: A
  survey,'' in {\em Control of Uncertain Systems} (D.~Hinrichsen and
  B.~M{\aa}rtensson, eds.), Progress in Systems and Control Theory,
  pp.~119--162, Boston, MA: {Birkh{\"a}user Boston}, 1990.

\bibitem{Hu.2003}
G.~Hu and E.~J. Davison, ``Real stability radii of linear time-invariant
  time-delay systems,'' {\em Systems {\&} Control Letters}, vol.~50, no.~3,
  pp.~209--219, 2003.

\bibitem{Zaborszky.1988}
J.~Zaborszky, G.~Huang, B.~Zheng, and T.-C. Leung, ``On the phase portrait of a
  class of large nonlinear dynamic systems such as the power system,'' {\em
  IEEE Transactions on Automatic Control}, vol.~33, no.~1, pp.~4--15, 1988.

\bibitem{Chiang.1996}
H.-D. Chiang and L.~Fekih-Ahmed, ``Quasi-stability regions of nonlinear
  dynamical systems: Theory,'' {\em IEEE Transactions on Circuits and Systems
  I: Fundamental Theory and Applications}, vol.~43, no.~8, pp.~627--635, 1996.

\bibitem{Rantzer.2001}
A.~Rantzer, ``A dual to {L}yapunov's stability theorem,'' {\em Systems {\&}
  Control Letters}, vol.~42, no.~3, pp.~161--168, 2001.

\bibitem{Sullivan.2015}
T.~J. Sullivan, {\em Introduction to uncertainty quantification}.
\newblock Cham and Heidelberg and New York and Dordrecht and London: Springer,
  2015.

\bibitem{Ott.2005}
W.~Ott and J.~A. Yorke, ``Prevalence,'' {\em Bulletin of the American
  Mathematical Society}, vol.~42, no.~03, pp.~263--291, 2005.

\bibitem{LaSalle.1976}
J.~P. LaSalle and Z.~Artstein, {\em The stability of dynamical systems}.
\newblock Philadelphia: SIAM, 1976.

\bibitem{Krasovskii.1964}
N.~N. Krasovskii, ``The approximation of a problem of analytic design of
  controls in a system with time-lag,'' {\em Journal of Applied Mathematics and
  Mechanics}, vol.~28, no.~4, pp.~876--885, 1964.

\bibitem{Delfour.1972}
M.~Delfour and S.~Mitter, ``Hereditary differential systems with constant
  delays. {I.} {G}eneral case,'' {\em Journal of Differential Equations},
  vol.~12, no.~2, pp.~213--235, 1972.

\bibitem{Bernier.1978}
C.~Bernier and A.~Manitius, ``On semigroups in {$\mathbb R^n \times L^p$}
  corresponding to differential equations with delays,'' {\em Canadian Journal
  of Mathematics}, vol.~30, no.~5, pp.~897--914, 1978.

\bibitem{Curtain.1995}
R.~F. Curtain and H.~Zwart, {\em An introduction to infinite-dimensional linear
  systems theory}.
\newblock New York, NY: Springer, 1995.

\bibitem{Krasovskii.1963}
N.~N. Krasovskii and J.~L. Brenner, {\em Stability of motion: Applications of
  Lyapunov's second method to differential systems and equations with delay}.
\newblock Stanford: {Stanford University Press}, 1963.

\bibitem{Minorsky.1948}
N.~Minorsky, ``Self--excited mechanical oscillations,'' {\em Journal of Applied
  Physics}, vol.~19, no.~4, pp.~332--338, 1948.

\bibitem{Lee.1982}
E.~Lee, S.~Neftci, and A.~Olbrot, ``Canonical forms for time delay systems,''
  {\em IEEE Transactions on Automatic Control}, vol.~27, no.~1, pp.~128--132,
  1982.

\bibitem{Daza.2017}
A.~Daza, A.~Wagemakers, and M.~A.~F. Sanju{\'a}n, ``Wada property in systems
  with delay,'' {\em Communications in Nonlinear Science and Numerical
  Simulation}, vol.~43, pp.~220--226, 2017.

\bibitem{Shampine.2001}
L.~F. Shampine and S.~Thompson, ``Solving {DDE}s in matlab,'' {\em Applied
  Numerical Mathematics}, vol.~37, no.~4, pp.~441--458, 2001.

\bibitem{Agrawal.2004}
V.~Agrawal, C.~Zhang, A.~D. Shapiro, and P.~S. Dhurjati, ``A dynamic
  mathematical model to clarify signaling circuitry underlying programmed cell
  death control in arabidopsis disease resistance,'' {\em Biotechnology
  Progress}, vol.~20, no.~2, pp.~426--442, 2004.

\bibitem{Bellen.2003}
A.~Bellen and M.~Zennaro, {\em Numerical methods for delay differential
  equations}.
\newblock Oxford: {Oxford Univ. Press}, 2003.

\bibitem{Shang.2009}
H.~Shang and J.~Xu, ``Delayed feedbacks to control the fractal erosion of safe
  basins in a parametrically excited system,'' {\em Chaos, Solitons {\&}
  Fractals}, vol.~41, no.~4, pp.~1880--1896, 2009.

\bibitem{Aguirregabiria.1987}
J.~M. Aguirregabiria and J.~R. Etxebarria, ``Fractal basin boundaries of a
  delay-differential equation,'' {\em Physics Letters A}, vol.~122, no.~5,
  pp.~241--244, 1987.

\bibitem{Losson.1993}
J.~Losson, M.~C. Mackey, and A.~Longtin, ``Solution multistability in
  first-order nonlinear differential delay equations,'' {\em Chaos}, vol.~3,
  no.~2, pp.~167--176, 1993.

\bibitem{Taylor.2007}
S.~R. Taylor and S.~A. Campbell, ``Approximating chaotic saddles for delay
  differential equations,'' {\em Phys. Rev. E}, vol.~75, no.~4, p.~046215,
  2007.

\bibitem{Smith.2011}
H.~Smith, {\em An introduction to delay differential equations with
  applications to the life sciences}.
\newblock New York: Springer, 2011.

\bibitem{Dudkowski.2016}
D.~Dudkowski, S.~Jafari, T.~Kapitaniak, N.~V. Kuznetsov, G.~A. Leonov, and
  A.~Prasad, ``Hidden attractors in dynamical systems,'' {\em Physics Reports},
  vol.~637, pp.~1--50, 2016.

\bibitem{Ilyashenko.2002}
Y.~Ilyashenko, ``Centennial history of {H}ilbert's 16th problem,'' {\em
  Bulletin of the American Mathematical Society}, vol.~39, no.~03,
  pp.~301--355, 2002.

\bibitem{Kazarinoff.1978}
N.~D. Kazarinoff, Y.-H. Wan, and P.~{van den Driessche}, ``Hopf bifurcation and
  stability of periodic solutions of differential-difference and
  integro-differential equations,'' {\em IMA Journal of Applied Mathematics},
  vol.~21, no.~4, pp.~461--477, 1978.

\bibitem{Hassard.1981}
B.~D. Hassard, N.~D. Kazarinoff, and Y.-H. Wan, {\em Theory and applications of
  Hopf bifurcation}.
\newblock Cambridge: {Cambridge University Press}, 1981.

\bibitem{Janssens.2010}
S.~G. Janssens, {\em On a normalization technique for codimension two
  bifurcations of equilibria of delay differential equations}.
\newblock Master thesis, {Utrecht University}, Utrecht, 2010.

\bibitem{Stepan.1991}
G.~St{\'e}p{\'a}n, ``Chaotic motion of wheels,'' {\em Vehicle System Dynamics},
  vol.~20, no.~6, pp.~341--351, 1991.

\bibitem{Molnar.2017}
T.~G. Molnar, Z.~Dombovari, T.~Insperger, and G.~St{\'e}p{\'a}n, ``On the
  analysis of the double hopf bifurcation in machining processes via centre
  manifold reduction,'' {\em Proceedings. Mathematical, physical, and
  engineering sciences}, vol.~473, no.~2207, p.~20170502, 2017.

\bibitem{Molnar.2018}
T.~G. Molnar, Z.~Dombovari, T.~Insperger, and G.~St{\'e}p{\'a}n, ``Bifurcation
  analysis of nonlinear time-periodic time-delay systems via
  semidiscretization,'' {\em International Journal for Numerical Methods in
  Engineering}, vol.~115, no.~1, pp.~57--74, 2018.

\bibitem{Sieber.2014}
J.~Sieber, K.~Engelborghs, T.~Luzyanina, G.~Samaey, and D.~Roose,
  ``{DDE-BIFTOOL} manual: Bifurcation analysis of delay differential
  equations,'' {\em arXiv:1406.7144}, 2014.

\bibitem{Roose.2007}
D.~Roose and R.~Szalai, ``Continuation and bifurcation analysis of delay
  differential equations,'' in {\em Numerical Continuation Methods for
  Dynamical Systems} (B.~Krauskopf, H.~M. Osinga, and J.~Gal{\'a}n-Vioque,
  eds.), pp.~359--399, Dordrecht: Springer, 2007.

\bibitem{Sprott.2007}
J.~C. Sprott, ``A simple chaotic delay differential equation,'' {\em Physics
  Letters A}, vol.~366, no.~4-5, pp.~397--402, 2007.

\bibitem{Lakshmanan.2010}
M.~Lakshmanan and D.~V. Senthilkumar, {\em Dynamics of Nonlinear Time-Delay
  Systems}.
\newblock Berlin, Heidelberg: {Springer-Verlag Berlin Heidelberg}, 2010.

\bibitem{Kuznetsov.1998}
Y.~A. Kuznetsov, {\em Elements of applied bifurcation theory}.
\newblock New York: Springer, 1998.

\bibitem{Broer.2011}
H.~W. Broer and F.~Takens, {\em Dynamical systems and chaos}.
\newblock New York: Springer, 2011.

\bibitem{Kuznetsov.2014b}
N.~V. Kuznetsov and G.~A. Leonov, ``Hidden attractors in dynamical systems:
  Systems with no equilibria, multistability and coexisting attractors,'' {\em
  IFAC Proceedings Volumes}, vol.~47, no.~3, pp.~5445--5454, 2014.

\bibitem{Efimov.2015}
D.~Efimov, J.~Schiffer, and R.~Ortega, ``Robustness of delayed multistable
  systems with application to droop-controlled inverter-based microgrids,''
  {\em International Journal of Control}, vol.~89, no.~5, pp.~909--918, 2015.

\bibitem{Balanov.2005}
A.~G. Balanov, N.~B. Janson, and E.~Sch{\"o}ll, ``Delayed feedback control of
  chaos: Bifurcation analysis,'' {\em Physical review. E, Statistical,
  nonlinear, and soft matter physics}, vol.~71, no.~1 Pt 2, p.~016222, 2005.

\bibitem{Oliva.1969}
W.~Oliva, ``Functional differential equations on compact manifolds and an
  approximation theorem,'' {\em Journal of Differential Equations}, vol.~5,
  no.~3, pp.~483--496, 1969.

\bibitem{Scholl.2019}
T.~H. Scholl, L.~Gr{\"o}ll, and V.~Hagenmeyer, ``Time delay in the swing
  equation: A variety of bifurcations,'' {\em Chaos: Interdiscip. J. Nonlinear
  Sci.}, vol.~29, no.~12, p.~123118, 2019.

\bibitem{Oxtoby.1980}
J.~C. Oxtoby, {\em Measure and category: A survey of the analogies between
  topological and measure spaces}.
\newblock New York: {Springer New York}, 1980.

\bibitem{Hunt.1992}
B.~R. Hunt, T.~Sauer, and J.~A. Yorke, ``Prevalence: {A} translation-invariant
  'almost every' on infinite-dimensional spaces,'' {\em Bulletin of the
  American Mathematical Society}, vol.~27, no.~2, pp.~217--239, 1992.

\end{thebibliography}

\section{Appendix}

\subsection{Definitions of Generalized Domains and Radii of Attraction}
For some purposes it is sufficient to describe a generalized domain of attraction in which "generically" convergence to the attractor occurs. 
There might be not attracted sets $M \not\subset \mathcal D_X$ of first (Baire) category in the interior of $\overline {\mathcal D_X}$ (see Chiang \cite{Chiang.1996}, Figure 2 for an ODE example). 
%Think of a saddle-type limit set $\omega$ in the neighborhood of the asymptotically stable zero equilibrium, like a saddle equilibrium or a saddle limit cycle. Not only the limit set itself could be a meager exclusion within $\overline {\mathcal D_X}$. 
%There is also its invariant stable manifold, which is analogously to the domain of attraction defined as
%\begin{align}
%W^s(\omega)\stackrel{\textrm{def}}=\{\phi\in X: \lim_{t\to\infty} \textrm {dist}(x_t(\cdot;\phi),\omega)=0\},
%\label{eq:defStableManifold}
%\end{align}
%where $\textrm {dist}(\cdot)$ denotes the distance in the induced metric.  
%\\
%\\
A numerical testing for convergence will typically not capture such sets. At the same time they form essential restrictions to analytical estimations that do not allow any exceptions (Figure \ref{fig:Measures}). Nevertheless, meager sets are also of less relevance in practical applications. 
As a consequence, the discussion in Section \ref{sec:Introduction} must be extended by requirements like
\\[\LL]
\begin{tabularx}{\linewidth}{lX}
{\itshape (R1') }& testing whether a certain initial function is within a generalized domain of attraction,
\\[\LL]
{\itshape (R2') }&proving generic convergence to the attractor, 
\\[\LL]
{\itshape (R3') }&disproving generic convergence to the attractor,
\\[\LL]
{\itshape (R4') }&comparing systems w.r.t. generalized domains of attraction.
\end{tabularx}
\\[\LL]
\\
Such generalizations of the domain of attraction are well known from ODEs.
The above described topological point of view is only one possibility.
It is taken by Zaborszky et al. \cite{Zaborszky.1988},
who introduce the topological concept of a quasi-stability boundary, which is defined as boundary of $\textrm{int} (\overline{\mathcal D_X})$. The main application is (delay-free) power system stability \cite{Chiang.1996,Chiang.2011}.
In contrast to this topological concept, there is also the measure theoretical one, which is based on convergence for almost every (a.e.) initial value in a domain. 
Rantzer et al. \cite{Rantzer.2001} introduce the approach of Lyapunov densities for almost global stability of equilibria in ODEs and thereby ignore non-attracted sets of Lebesgue measure zero. Measure theoretic considerations can further be transferred to a probability space such that (R1') means convergence to the attractor with probability one / almost sure convergence (equivalently). 
It should be noted that the topological definition of a set of first category / a meager set on the one hand and the measure theoretical definition of a  a null set / a set of measure zero on the other hand are not equivalent. Neither includes the other \cite{Oxtoby.1980}. The authors of the present paper assume that it is an open question whether for subsets of the boundary $\partial \mathcal D_X$ equivalence can be assumed since it would require more knowledge about possible fractal structures. 
Below, we propose definitions of generalized domains of attraction in time\rev{-}delay systems. Thereby, we use the concept of prevalence by Hunt \cite{Hunt.1992} to address the lack of a Lebesgue measure in infinite dimensional spaces (Remark \ref{rem:noLebesgueMeasure}).
\begin{definition} [Generalized Concepts for Domains of Attraction] \label{def:gen_domains}
i) The {\itshape quasi-stability domain} is defined as 
\begin{align}
\mathcal D_X^{qs}\stackrel{\textrm{def}}=\textrm{int} (\overline{\mathcal D_X}), \label{eq:qs_domain}
\end{align}
where $\overline{\mathcal D_X}$ denotes the closure of the domain of attraction.  \\
ii)  An {\itshape almost-everywhere stability domain} is any set $\mathcal D_X^{ae}\supseteq \mathcal D_X$ with 
\begin{align}
\phi \in \mathcal D_X \textrm{ for a.e. } \phi \in \mathcal D^{ae} ,
\end{align}
where a.e. is for infinite dimensional spaces understood in the sense of prevalence \cite{Hunt.1992}.
\end{definition}
\noindent Obviously, these generalizations are not smaller than the actual domain of attraction, i.e.\ $\mathcal D_X\subseteq \mathcal D_X^{qs}$ as well as $\mathcal D_X\subseteq \mathcal D_X^{ae}$. 
We denote the corresponding radius of attraction by 
\begin{align}
 R^{gen}_X=\sup\{r > 0: \mathcal B_X(r)\subseteq \mathcal D^{gen}_{X}\}.
\end{align}
with $gen\in \{qs,ae\}$.

\end{document}